\documentclass[a4paper,french,12pt,dvipsnames]{article}
\usepackage[utf8]{inputenc}
\usepackage{amsmath}
\usepackage{dsfont}
\usepackage{mathrsfs}
\usepackage{tikz}
\usetikzlibrary{matrix}
\usepackage[absolute,overlay]{textpos}
\usepackage{cases}

\DeclareMathAlphabet{\mathdutchcal}{U}{dutchcal}{m}{n}
\SetMathAlphabet{\mathdutchcal}{bold}{U}{dutchcal}{b}{n}
\DeclareMathAlphabet{\mathdutchbcal}{U}{dutchcal}{b}{n}

\DeclareMathOperator{\Tr}{Tr}

\DeclareMathOperator{\Pf}{Pf}
\DeclareMathOperator{\Ad}{Ad}
\DeclareMathOperator{\ad}{ad}

\usepackage{enumitem}
\usepackage{etoolbox}

\usepackage{graphicx}               % Necessary to use \scalebox

\usepackage{chngcntr}

\usepackage{slashed}
\usepackage{pgf} 
\usepackage{stmaryrd}
\usepackage{hyperref}
\usepackage{mathtools}

\usepackage{etex}

\usepackage{grffile}

\usepackage{accents}
\usepackage{bm}
\usepackage{enumitem}
\usepackage{amssymb}
\usepackage{extarrows}
\usepackage{calrsfs}
\usepackage{amsthm}

\usepackage{geometry}
\usepackage[T1]{fontenc}
\usepackage{calligra}
\DeclareMathAlphabet{\mathpzc}{OT1}{pzc}{m}{it}
\usepackage{array,multirow,makecell}
\setcellgapes{1pt}
\makegapedcells
\newcolumntype{R}[1]{>{\raggedleft\arraybackslash }b{#1}}
\newcolumntype{L}[1]{>{\raggedright\arraybackslash }b{#1}}
\newcolumntype{C}[1]{>{\centering\arraybackslash }b{#1}}

\newtheorem{prop}{Proposition}[section]

\theoremstyle{definition}
\newtheorem{definition}{Definition}[section]

\theoremstyle{remark}
\newtheorem*{remark}{Remark}

\newcommand\numberthis{\addtocounter{equation}{1}\tag{\theequation}}

\newcommand{\p}{\mathfrak{p}}

\newcommand{\h}{\mathfrak{h}}
\newcommand{\g}{\mathfrak{g}}
\newcommand{\m}{\mathfrak{m}}

\newcommand{\cm}{\mathcal{M}}

\newcommand{\GamTtM}{\Gamma(T\widetilde{\mathcal{M}})}

\newcommand{\tg}{\tilde{g}}

\newcommand{\tk}{\tilde{k}}

\newcommand{\tl}{\tilde{\ell}}

\newcommand{\varpih}{\varpi_\mathfrak{h}}

\newcommand{\varpiz}{\varpi_{0}}

\newcommand{\varpip}{\varpi_\p}

\newcommand{\varpim}{\varpi_\m}

\newcommand{\varpihmod}[2][ ]{\varpi^{ #1}_{\h \:\,#2}} 

\newcommand{\Amod}[2][ ]{A^{#1}_{~\,\,#2}}

\newcommand{\betamod}[2][ ]{\beta^{\,#1}_{~\,\,#2}} 

\newcommand{\varpipmod}[2][ ]{\varpi^{ ~\,#1}_{\p\,\,#2}}

\newcommand{\betainvmod}[2][ ]{\beta^{-1 \, #1}_{~~\,\,\; \; \;#2}}

\newcommand{\Rmod}[2][ ]{R^{#2}_{\;\; \; #1}} 

\newcommand{\Rzmod}[2][ ]{R^{#2}_{0\; \; #1}} 

\newcommand{\rR}{\mathrm{R}}

\newcommand{\Tmod}[2][ ]{T^{#2}_{\;\;  #1}}

\newcommand{\bOm}{\bar{\Omega}}

\newcommand{\bOmh}{\bar{\Omega}_{\h}}

\newcommand{\bOmm}{\bar{\Omega}_{\m}}

\newcommand{\bOmp}{\bar{\Omega}_{\p}}

\newcommand{\tS}{\tilde{S}}

\makeatletter
\def\lefteqno{\tagsleft@true}\def\righteqno{\tagsleft@false}
\makeatother 

\makeatletter
\newenvironment{myalign*}{\ifvmode\else\hfil\null\linebreak\fi
	\hspace*{-\leftmargin}\minipage\textwidth
	\setlength{\abovedisplayskip}{0pt}%
	\setlength{\abovedisplayshortskip}{\abovedisplayskip}%
	\start@align\@ne\st@rredtrue\m@ne}%
{\endalign\endminipage\linebreak}
\makeatother

\providecommand{\keywords}[1]
{
	\small	
	\textbf{Keywords :} #1
}

\makeatletter
\let\save@mathaccent\mathaccent
\newcommand*\if@single[3]{%
	\setbox0\hbox{${\mathaccent"0362{#1}}^H$}%
	\setbox2\hbox{${\mathaccent"0362{\kern0pt#1}}^H$}%
	\ifdim\ht0=\ht2 #3\else #2\fi
}
\newcommand*\rel@kern[1]{\kern#1\dimexpr\macc@kerna}
\newcommand*\widebar[1]{\@ifnextchar^{{\wide@bar{#1}{0}}}{\wide@bar{#1}{1}}}
\newcommand*\wide@bar[2]{\if@single{#1}{\wide@bar@{#1}{#2}{1}}{\wide@bar@{#1}{#2}{2}}}
\newcommand*\wide@bar@[3]{%
	\begingroup
	\def\mathaccent##1##2{%
		\let\mathaccent\save@mathaccent
		\if#32 \let\macc@nucleus\first@char \fi
		\setbox\z@\hbox{$\macc@style{\macc@nucleus}_{}$}%
		\setbox\tw@\hbox{$\macc@style{\macc@nucleus}{}_{}$}%
		\dimen@\wd\tw@
		\advance\dimen@-\wd\z@
		\divide\dimen@ 3
		\@tempdima\wd\tw@
		\advance\@tempdima-\scriptspace
		\divide\@tempdima 10
		\advance\dimen@-\@tempdima
		\ifdim\dimen@>\z@ \dimen@0pt\fi
		\rel@kern{0.6}\kern-\dimen@
		\if#31
		\overline{\rel@kern{-0.6}\kern\dimen@\macc@nucleus\rel@kern{0.4}\kern\dimen@}%
		\advance\dimen@0.4\dimexpr\macc@kerna
		\let\final@kern#2%
		\ifdim\dimen@<\z@ \let\final@kern1\fi
		\if\final@kern1 \kern-\dimen@\fi
		\else
		\overline{\rel@kern{-0.6}\kern\dimen@#1}%
		\fi
	}%
	\macc@depth\@ne
	\let\math@bgroup\@empty \let\math@egroup\macc@set@skewchar
	\mathsurround\z@ \frozen@everymath{\mathgroup\macc@group\relax}%
	\macc@set@skewchar\relax
	\let\mathaccentV\macc@nested@a
	\if#31
	\macc@nested@a\relax111{#1}%
	\else
	\def\gobble@till@marker##1\endmarker{}%
	\futurelet\first@char\gobble@till@marker#1\endmarker
	\ifcat\noexpand\first@char A\else
	\def\first@char{}%
	\fi
	\macc@nested@a\relax111{\first@char}%
	\fi
	\endgroup
}
\makeatother

\setlist[enumerate]{label=\thesection.\arabic{*},resume}

\preto\section{%
	\restartlist{enumerate}%
}

\let\oldsection\section% Store \section
\renewcommand{\section}{% Update \section
	\renewcommand{\theequation}{\thesection.\arabic{equation}}% Update equation number
	\oldsection}% Regular \section
\let\oldsubsection\subsection% Store \subsection
\renewcommand{\subsection}{% Update \subsection
	\renewcommand{\theequation}{\thesubsection.\arabic{equation}}% Update equation number
	\oldsubsection}% Regular \subsection
\counterwithin*{equation}{section}
\counterwithin*{equation}{subsection}
\date{}
\begin{document}
\author{Jean Thibaut\,\footnote{jeanthibautt@outlook.fr} \\
	{\normalsize Centre de Physique Théorique}\\
	{\normalsize Aix Marseille Univ, Université de Toulon, CNRS, CPT, Marseille, France.}}

\title{Emergent, spacetime-dependent dark energy and gravitational coupling from moving geometries}

\maketitle

		\begin{abstract}
		    We introduce the notion of moving Cartan geometries described by quotients of Lie groups and Lie algebras with spacetime dependant structure constants and construct associated deformed topological gauge action functionals for Lorentzian (including dS and AdS) and Lorentz$\times$Weyl moving geometries. The actions feature a generalization of the Nieh-Yan topological term for a varying coupling constant. 
		    We compute the equations of motion of the  gauge $+$ matter actions and show that they dictate at each spacetime point the geometry, leading to both a spacetime-dependent source of dark energy and gravitational coupling, described by combinations of scalars built from spacetime curvature, torsion and the matter content of the theory. The action becomes asymptotically topological when the gauge action contribution to dark energy vanishes.
		\end{abstract}
    
	\keywords{Topological gravity, Homogeneous/symmetric spaces, Cartan geometry, Characteristic classes, Varying dark energy, Varying gravitational coupling}
		
    \newpage
	
	%\tableofcontents
	%\newpage
	%\listoffigures

		\section{Introduction}
		
		In the literature, promoting coupling constants to scalar fields has already been explored via theories like Chern-Simons gravity \cite{alexander_chern-simons_2009} or by promoting the Barbero-Immirzi parameter to a scalar field \cite{taveras_barbero-immirzi_2008}.
		Other approaches \cite{hamber2011scale} also describe a running Newton constant $G$ motivated by the renormalization group.
		The same goes for descriptions of a varying cosmological term (replacing the cosmological constant) $\Lambda$ in terms of a scalar field (e.g. \cite{alexander_quantum_2019,alexander_zero-parameter_2019}) acting as a source of dark energy (see \cite{copeland_dynamics_2006,li2013dark} for reviews). A mix of both gravitational coupling and cosmological term varying in time has also been described in \cite{sengupta2025cosmological}, see \cite{uzan2011varying} for a review of similar attempts. Recent developments regarding the tensions in cosmology \cite{tang_uniting_2025,collaboration_desi_2025} seem to push in the direction of an evolving form of dark energy, with a transfer from dark energy to dark matter at late cosmological times. Thus, renewing interest in the previous type of approaches consisting of promoting constants to scalar fields.
		The goal of the present paper is to apply this kind of procedure to Cartan geometry, directly at the level of Lie groups, Lie algebras and their structure constants. The new scalar fields are then entirely determined at each point by enforcing the action principle without needing any additional input.
		
		Let $\mathcal{M}$ be a manifold described by a Cartan geometry $G/H$, where $SO(3,1) \subset H\subset G$ are Lie groups and $\mathfrak{so}(3,1) \subset \h \subset \g$ are their respective Lie algebras. We are interested in the case where $\g = \h \oplus \m$ is a symmetric Lie algebra since the structure constant corresponding to the Lie bracket  $[\m,\m] \subset \h$ is directly linked to a bare cosmological constant $\Lambda_0$.
		In this article we will consider the geometries described in \cite{thibaut2024gravity}
        and generalize them by promoting to scalar fields the structure constants which gave rise to the bare cosmological constant $\Lambda_0$ in the aforementioned article.
		The gauge action we consider is the one from \cite{thibaut2024gravity}, namely (the second equality holds for $\dim(\mathcal{M})=4$):
		\begin{align} \label{eq:action_start}
            S_G [\varpi]  = &\int_\mathcal{M} \big(r P(\bOm) + e \Pf  (\dfrac{F}{2\pi}) + y \det ( \mathds{1} + \dfrac{\bOmh}{2\pi})\big)  \notag\\
            & \\
            = \dfrac{1}{8\pi^2} \!\!\int_{\mathcal{M}} & \Bigl(
	    	\dfrac{ e }{4\pi^2} 
		    \varepsilon_{abcd} F^{ab} \wedge F^{cd} 
		       - r \Tr (\bOm^2)
		    - y \Tr (\bOmh^2)
		    \Bigr) \nonumber
        \end{align}
	where the first term is the Pontryagin number of the manifold related to the $\g$-valued Cartan curvature $\bOm = d\varpi + \varpi\wedge \varpi$ originating from the Cartan connection $\varpi$ while the two other terms are the invariant symmetric polynomials corresponding to the Pfaffian $\Pf (\dfrac{F}{2\pi})$ and the determinant $\det ( \mathds{1} + \dfrac{\bOmh}{2\pi})$. They are respectively evaluated on $F$ (the $\mathfrak{so}(3,1)$-valued part of the curvature $\bOmh$) and $\bOmh$ ($\h$-valued) parts of the Cartan curvature. The resulting action is Lorentz invariant and is the deformation of a topological gauge theory.

    The next two sections will be devoted to define the moving geometries by promoting $\Lambda_0$ to a scalar field $\chi$, compute action \eqref{eq:action_start} and build matter actions before finally deriving the equations of motion (EOM's) of the gauge+matter actions for respectively the Lorentz and Lorentz$\times$Weyl (conformal) moving geometries. We will show the EOM's in fact dictate the geometries at each point and impose a spacetime dependent gravitational coupling $G \propto \dfrac{1}{\chi} \text{ or } \chi$ (depending on the choice of matter action) and a source of dark energy $\Lambda_G \propto \chi$. In general, $\chi$ depends on shell on scalars of the spacetime curvature, torsion and the matter content of the theory.
		
	In one of the simplest cases we retrieve what can be identified as a Ricci dark energy model similar to the one treated in \cite{gao_holographic_2009} such that $\Lambda_G = \dfrac{1}{4} \rR$ with parameter $\alpha = \dfrac{1}{2}$. The difference being that we end up with a spacetime dependent gravitational coupling $G =G_0 \dfrac{\rR(x_0)}{\rR(x)}$, such that we retrieve Newton's constant $G_0$ at a point $x_0\in \cm$ with $\mathrm{R}$ the Ricci scalar curvature.

		\section{Moving Lorentzian Cartan geometry}
		\label{Moving Lorentzian Cartan geometry}
        Our examples of moving geometries rely on the concept of mutation introduced in \cite[see p.218]{sharpe_differential_1997}:
	\begin{definition} \label{defmutation}
		Let $(\g,\h)$ and $(\g',\h)$ be two model geometries.
		A mutation map corresponds to an $\Ad(H)$ module isomorphism $\mu : \g \rightarrow \g'$ ({\sl i.e.} $\mu(\Ad_h(u)) = \Ad_h(\mu(u))\, \forall u\in\g$) satisfying:
		\begin{flalign*}
			(i) &\ \mu_{|\h} = \mathrm{id}_\h \\
			(ii) &\ \bigl[ \mu(u) , \mu(v) \bigr] = \mu \bigl([ u,v ] \bigr) \text{ mod } \h, \forall u,v \in \g.
		\end{flalign*}
		The model geometry $(\g',\h)$ then corresponds to the mutant of the model geometry $(\g,\h)$ with the same group $H$.
	\end{definition}
        
		\subsection{Moving Lorentzian spacetimes as mutated $G'/H$ geometries}
        \label{Moving Lorentzian gauge theory}

	Let $\h=\mathfrak{so}(3,1)$ be the Lorentz Lie algebra.
	The Cartan geometry we want to study corresponds to the mutated geometry $(\g,\h)$ described in \cite{thibaut2024gravity}, with Lie group element $g\in G$ satisfying:
	\begin{align}
	    \label{mutation Lorentz}
		g^TNg=N ~,~ 
		N= \begin{pmatrix}
			\eta	& 0 \\
			0 & - \dfrac{k_1}{k_2}
		\end{pmatrix}
	\end{align}
	and Lie algebra element:
	\begin{align}
		\varphi =
		\begin{pmatrix}
			\varphi_\h	&  k_1 \varphi_\p \\
			k_2 \varphi_{\bar{\p}} & 0
		\end{pmatrix} \in \g
	\end{align}
	where this time the mutation parameters $k_1(x)$ and $k_2(x)$ are functions (0-form scalar fields) depending on the points of the manifold $x\in \cm$. Such geometries, locally described by algebras with varying structure constants will be called moving geometries.
	This Lie algebra can be seen as originating from the mutation of the $ISO(3,1)/SO(3,1)$ Cartan geometry first introduced in \cite{thibaut2024gravity}. This type of mutation with parameters described by scalar fields may be used in the future to retrieve other theories than the ones described below and based on quadratic extensions of Lie algebras including generalizations of the ones described in \cite{cattaneo2024gravity}.
	
	\begin{remark}
		If $k_1/k_2=1$ we retrieve the AdS Lie algebra, while if $k_1/k_2=-1$ we get the dS Lie algebra.
	\end{remark}
	
	We write $\g = \h \oplus \m$ the symmetric decomposition of $\g$ with $\m =
	\mathbb{R}^{3,1}$. The corresponding matrix representation is:
	\begin{align}
		\varphi &= 
		\begin{pmatrix}
			\varphi_\h	& k_1 \varphi_\p \\
			k_2 \varphi_{\bar{\p}} & 0
		\end{pmatrix} 
		=
		\begin{pmatrix}
			\dfrac{1}{2}\varphi_\h^{ab} J_{ab}	& k_1 \varphi_\p^a e_a \\
			k_2 \varphi_{p}^a \bar{e}_a & 0
		\end{pmatrix}
		=
		\dfrac{1}{2} \varphi_\h^{ab}
		\begin{pmatrix}
			J_{ab}	& 0 \\
			0 & 0
		\end{pmatrix}
		\oplus
		\varphi_\m^a
		\begin{pmatrix}
			0	& k_1 e_a \\
			k_2 \bar{e}_a & 0
		\end{pmatrix} \\
		&= \dfrac{1}{2} \varphi_\h^{ab} J_{ab} \oplus \varphi_\m^a M_a
	\end{align}
	with $\varphi^a_\m = \varphi^a_\p$ and 	
	$M_a = \begin{pmatrix}
		0	& k_1 e_a \\
		k_2 \bar{e}_a & 0
	\end{pmatrix}$ the generators of $\m$. The generators $J_{ab}$ of $\h$ and the basis elements $e_a$ and $\bar{e}_a$ obey the following relations:
	\begin{align*}
		\bar{e}_a (e_b) = \eta_{ab} \quad , \quad J_{ab} = e_a \bar{e}_b - e_b \bar{e}_a .
	\end{align*}

	\smallskip	
	Let $\varphi,\varphi' \in \g$,  their Lie bracket reads: 
	\begin{align*}
		[\varphi,\varphi'] &= 
		\begin{bmatrix}
			\begin{pmatrix}
				\varphi_\h	& k_1 \varphi_\p \\
				k_2 \varphi_{\bar{\p}} & 0
			\end{pmatrix},
			\begin{pmatrix}
				\varphi'_\h	& k_1 \varphi'_\p \\
				k_2 \varphi'_{\bar{\p}} & 0
			\end{pmatrix}
		\end{bmatrix} 
		=	\begin{pmatrix}
			[\varphi_\h , \varphi'_h] + k_1k_2 (\varphi_\p \varphi'_{\bar{\p}} -  \varphi'_\p \varphi_{\bar{\p}}) & k_1 (\varphi_\h \varphi'_\p - \varphi'_\h \varphi_\p) \\
			k_2 (\varphi_{\bar{\p}} \varphi'_\h - \varphi'_{\bar{\p}} \varphi_\h) & 0
		\end{pmatrix} \\
		& = \bigl( \dfrac{1}{8} \varphi_\h^{ab} \varphi_h^{\prime cd}  C_{ab,cd}^{ef} + \dfrac{k_1k_2}{2} ( \varphi_\p^{e} \varphi_\p^{\prime f} -  \varphi_\p^{\prime e} \varphi_\p^{f} ) \bigr) J_{ef}
		\oplus 
		\dfrac{1}{2} (  \varphi_\h^{ab} \varphi_\p^{\prime c} - \varphi_\h^{\prime ab} \varphi_\p^{c} ) ( \eta_{bc} \delta^r_a - \eta_{ac} \delta^r_b  ) M_r \numberthis
	\end{align*}
	where we have used $[J_{ab},M_c]=\eta_{bc} M_a - \eta_{ac} M_b$ and $[M_a,M_b]=k_1k_2\,J_{ab}$.\footnote{The product of the mutation parameters $k_1$ and $k_2$ can in fact be related to a deformation parameter $\alpha$ as given in \cite[eq.(2.1)]{chirco_gravity_2025} and may provide a good framework to encode the notion of Cartan khronon field described in \cite{koivisto2025spin}.}

	\subsection[Cartan connection]{Cartan connection}
	
	In order to stick with notations in the literature we now consider $k_1 = \dfrac{k}{\ell}$ and $k_2 = \dfrac{k'}{\ell}$, with $k,k'$ and $\ell$ $0$-form scalar fields.
	
	The associated Cartan connection is:
	\begin{align}
		\varpi = \varpih \oplus \varpim
		=
		\begin{pmatrix}
			\dfrac{1}{2}A^{ab} J_{ab}	& \dfrac{k}{\ell}\beta^a e_a \\
			\dfrac{k'}{\ell}\beta^a \bar{e}_a & 0
		\end{pmatrix}
		= \begin{pmatrix}
			A & \dfrac{k}{\ell}\beta \\
			\dfrac{k'}{\ell} \bar{\beta} & 0
		\end{pmatrix}
	\end{align}
	where $\varpih^{ab}=\varpihmod[ab]{\mu} dx^\mu = \Amod[ab]{\mu} dx^\mu $ corresponds
	to the spin connection and $\varpip^a = \betamod[a]{\mu}dx^\mu$ corresponds to the tetrad while $\varpim = \beta^a M_a$ is the soldering form and $\bar{\beta}=\beta^{\text{T}}\eta$. In the Cartan connection $k,k'$ and $\ell$ have now the physical effect of rescaling the tetrad. We will see they play a similar role to the dS/AdS radius with the difference that $k,k'$ and $\ell$ depend on spacetime points and also impact the gravitational coupling $G$.
	
	The curvature associated to the covariant derivative and the Cartan connection is:
	\begin{align}
		\bOm & = d \varpi + \dfrac{1}{2} [\varpi,\varpi] = \bOmh \oplus \bOmm \oplus \dfrac{dk}{\ell} \beta \oplus \dfrac{dk'}{\ell} \bar{\beta}  \\
		& = \bigl(d \varpih + \dfrac{1}{2} [\varpih,\varpih] + k k' \xi\bigr) \oplus \dfrac{dk}{\ell} \beta \oplus \dfrac{dk'}{\ell} \bar{\beta}
		\\
		& \qquad \oplus 
		\bigl( d \varpip^r + \dfrac{1}{4} ( \varpihmod[ab]{\mu} \varpipmod[c]{\nu} - \varpihmod[ab]{\nu} \varpipmod[c]{\mu} ) ( \eta_{bc} \delta^r_a - \eta_{ac} \delta^r_b  ) dx^\mu \wedge dx^\nu - \dfrac{d\ell}{\ell} \wedge \varpip^r \bigr) \otimes M_r \notag \\
		& = \bigl( \dfrac{1}{2} \partial_{\mu} \Amod[rs]{\nu} + \dfrac{1}{32}  \Amod[ab]{\mu} \Amod[cd]{\nu} C_{ab,cd}^{rs} 
		+ \dfrac{kk'}{4} \xi^{rs}_{~~\mu \nu} \bigr) dx^\mu \wedge dx^\nu \otimes J_{rs} \oplus \dfrac{dk}{\ell} \beta \oplus \dfrac{dk'}{\ell} \bar{\beta} \nonumber \\
		& \qquad 
		\oplus 
		\bigl( \partial_{\mu} \betamod[r]{\nu} + \dfrac{1}{2} \Amod[ab]{\mu} \betamod[c]{\nu} ( \eta_{bc} \delta^r_a - \eta_{ac} \delta^r_b ) - \dfrac{\partial_\mu\ell}{\ell} \varpipmod[r]{\mu} \bigr) dx^\mu \wedge dx^\nu \otimes M_r
	\end{align}
	where $\xi = \dfrac{1}{\ell^2} \beta \wedge \bar{\beta}= \dfrac{1}{2} \xi^{ab} J_{ab} = \dfrac{1}{4} \xi^{ab}_{~~\mu\nu} dx^\mu \wedge dx^\nu \otimes J_{ab} $ is an additional term due to the symmetric Lie algebra structure, 
	with $ \xi^{ab} = \dfrac{1}{\ell^2} \varpip^{a} \wedge \varpip^{b} $ 
	and $ \xi^{ab}_{~~\mu\nu} = \dfrac{1}{\ell^2} (\varpipmod[a]{\mu} \varpipmod[b]{\nu} - \varpipmod[a]{\nu} \varpipmod[b]{\mu})$.

	\begin{remark}
		One can observe that the curvature gives rise to the terms: 
		$$\dfrac{dk}{\ell} \beta \oplus \dfrac{dk'}{\ell} \bar{\beta}$$
		These terms do not in general stay inside the Lie algebra $\g$ and instead make the curvature $\mathfrak{sl}(5)$-valued. The sufficient and necessary condition for the curvature to exclusively have values in $\g$ is:
		\begin{align}
			\label{condition stay algebra}
			k' dk = k dk'
		\end{align} 
		In the non degenerate case $k_1/k_2 \neq 0$ this condition can be understood as a metric compatibility condition $$D_\varpi N = dN + \varpi^T N + N \varpi =0 $$
		In this setting the non-metricity is related to the total five-dimensional internal structure and should not be confused with non-metricity of the spacetime metric.
		For the moment we choose not to impose this condition.
	\end{remark}

	We identify the curvature and torsion of Einstein-Cartan gravity:
	\begin{align}
		R &= \bOmh - k k' \xi = d A + \dfrac{1}{2} [A,A] \\
		T &= 
		d\beta + A\wedge\beta 
	\end{align}
	Let $\tg = \beta^*\eta$ be the spacetime metric associated to the tetrad $\beta$.
	Torsion can be further divided into tensor $q_{\lambda\mu\nu}$, axial $S^\rho = \dfrac{1}{6} \varepsilon^{\rho \mu \nu \sigma} T_{\mu\nu\sigma}$ and trace $t_\mu = T^\nu_{~\mu\nu}$ components with the expression $T_{\alpha\mu\nu} = \dfrac{1}{3} ( \tg_{\alpha\nu} t_\mu - \tg_{\alpha\mu} t_\nu )  - \varepsilon_{\mu\nu\alpha\beta} S^\beta + q_{\alpha\mu\nu}$. 
	
	The Bianchi identities are $D\bOm = d\bOm + [\varpi,\bOm] =0$, where $[\cdot,\cdot]$ is the matrix commutator. They are satisfied even if one doesn't impose \eqref{condition stay algebra}.

	If one wants to keep the topological properties of the actions presented in \cite{thibaut2024gravity} the Chern-Weil theorem needs to hold and therefore we need the Bianchi identities.
    The Bianchi identities are of course verified if $k$, $k'$ and $\ell$ are constants.

	We will now make use of this expression of the curvature to compute the gauge action \eqref{eq:action_start} of earlier.
		
		%%%%%%%%%%%%%%%%%%%%%%%%%%%%%%%%%%%%%%%%%%%%%%%%%%%%%%%%%%%%%%%%%%%%%%%%%%%%%%%%%%%%
		
	\subsection[Moving Lorentzian geometry actions]{Moving Lorentzian geometry actions}
	\label{Moving Lorentzian geometry actions}
	
	The trace $\Tr : \g \times \g \rightarrow \mathbb{R}$
	will be used to compute the traces appearing in the determinants of action \eqref{eq:action_start}.
	Here $\Tr(\varphi,\varphi') = - \eta_{ac} \eta_{bd}\varphi_\h^{ab} \varphi_\h^{cd} + \dfrac{k k'}{\ell^2} \eta_{ab} (\varphi_{\p}^a \varphi_{\bar{\p}}^{\prime b} + \varphi_{\p}^{\prime a} \varphi_{\bar{\p}}^b)$.
	\medskip
	
	For this geometry, action \eqref{eq:action_start} yields:
	\begin{align*}
		S_G[\varpi] 
		%%%%%%%%%%%%%%%%%%%%%%%%%%%%%%%%%%%%%%%%%%%%
		& = \int_{\mathcal{M}} \biggl(
		\dfrac{e }{2(4\pi)^2} \varepsilon_{abcd} \bigl(  R^{ab} \wedge R^{cd} 
		+ 2kk' R^{ab} \wedge \xi^{cd} 
		+ (kk')^2 \xi^{ab} \wedge \xi^{cd}
		\bigr) \numberthis \\
		& \qquad \quad + \dfrac{r+y}{8 \pi^2} \bigl( R^{ab} \wedge R_{ab} + 2kk'  R^{ab} \wedge \xi_{ab}
		\bigr) \\
		%%%%%%%%%%%%%%%%%%%%%%
		& \qquad \quad - \dfrac{r }{8 \pi^2} \bigl( \dfrac{2kk'}{\ell^2} T^a \wedge T_a
		+ \dfrac{2kk'}{\ell^2} (\dfrac{dk}{k} + \dfrac{dk'}{k'} - 2\dfrac{d\ell}{\ell}) \wedge \beta^a \wedge T_a + \dfrac{2kk'}{\ell^4} (d\ell)^2 \wedge \overbrace{\beta^a \wedge \beta_a}^{=0}
		\bigr)
		\biggr)
	\end{align*}
	
	Identifying the different terms, we find:
	\begin{align*}
		\label{action Lorentz}
		S_G[\varpi] 
		&= \int_{\mathcal{M}} \Bigl(
		\overbrace{ \dfrac{kk' }{4\pi^2\ell^2} \big( \dfrac{e}{4} \underbrace{R^{ab} \wedge \beta^c\wedge\beta^d \varepsilon_{abcd}}_{Palatini} + y R^{ab} \wedge \beta_a\wedge\beta_b \big)}^{Holst}
		+ \overbrace{\dfrac{e (kk')^2}{32\pi^2\ell^4}  \beta^a\wedge\beta^b \wedge \beta^c\wedge\beta^d \varepsilon_{abcd}}^{Bare~Cosmological~function} \nonumber
		\\
		%%%%%%%%
		& \qquad \quad + \overbrace{\dfrac{r +y}{8 \pi^2} R^{ab} \wedge R_{ab}}^{Pontrjagin}
		+ \overbrace{\dfrac{e }{2(4\pi)^2} R^{ab} \wedge R^{cd} \varepsilon_{abcd}}^{Euler} 
		\overbrace{- \dfrac{r kk' }{4 \pi^2 \ell^2} (T^a \wedge T_a - R^{ab} \wedge \beta_a\wedge\beta_b)}^{Nieh-Yan}
		\Bigr) \nonumber \\
		&  \qquad \qquad \overbrace{- \dfrac{rkk'}{4\pi^2\ell^2} (\dfrac{dk}{k} + \dfrac{dk'}{k'} - 2 \dfrac{d\ell}{\ell}) \wedge \beta^a \wedge T_a}^{Kinetic~term~for~k,~k',~and~\ell}
		\Bigr) \numberthis
	\end{align*}
	As an example we can recover MacDowell-Mansouri (MM) gravity (see \cite{wise_macdowell-mansouri_2010} for a treatment of MM gravity in terms of Cartan geometry) by requiring $r=y=0$ and imposing $k,k',\ell$ to be constants.

	We further find that, when combined, the Nieh–Yan and kinetic terms constitute a generalization of the Nieh–Yan topological invariant to the case of a moving geometry with spacetime-dependent scalar fields $k$, $k'$, and $\ell$. As will be shown in Section \ref{Lorentzian Equations of motion}, their combined contribution does not modify the equations of motion. Consequently, the terms appearing in the last two lines remain purely topological, even in the presence of a matter action. It is however important to note that in a path integral formulation of quantum gravity, these last terms, including the new kinetic term, could still yield non trivial effects when summing over all possible configurations as explained in \cite[p.2]{rezende_4d_2009}.
    
	Moving on to build matter actions for this geometry we first define metrics $h$ on $\g$ and $g$ on $\cm$ with $*$ the Hodge star operator associated to $g$. 
	Choosing the symmetric decomposition $ \g= \h \oplus \overbrace{\p \oplus \bar{\p}}^\m $ with $\m = \mathbb{R}^{3,1}$ and $\h = \mathfrak{so}(3,1)$ we compute the Killing metric: 
	
	$$K_\g (\varphi,\varphi') = 3 \Tr(\varphi_\h , \varphi'_\h) + 3 \dfrac{kk'}{\ell^2} \bigl(\eta(\varphi_p , \varphi^{\prime}_{\bar{\p}}) + \eta(\varphi'_p , \varphi_{\bar{\p}}) \bigr)
	= \overbrace{- 3 \eta_{ac} \eta_{bd} \varphi_\h^{ab} \varphi_h^{\prime cd}}^{K_\h (\varphi_\h,\varphi_\h')} + \overbrace{3 \dfrac{kk'}{\ell^2} \eta_{ab} (\varphi_\p^{\prime a} \varphi_{\bar{\p}}^{b} + \varphi_\p^a \varphi_{\bar{\p}}^{\prime b})}^{K_\m (\varphi_\m,\varphi_\m')}  $$
	giving us the metrics: 
	\begin{align}
		& h(\varphi,\varphi') = \kappa K_\g (\varphi,\varphi') \\
		&g (X,Y)  = \zeta (\varpim^*h) (X,Y) = \kappa \zeta (\varpim^*K_\m) (X,Y) = \dfrac{6kk' \zeta \kappa }{\ell^2} \eta \bigl( \beta(X) , \beta(Y) \bigr) = \dfrac{6kk' \zeta \kappa }{\ell^2} \tilde{g} (X,Y) .
	\end{align}
	where $\kappa(x)$ and $\zeta(x)$ are functions depending on the point of the manifold $x \in \cm$. 
	
	For these metrics, non-degeneracy requires $\kappa,\zeta,k,k'\neq 0$ $\forall x \in  \cm$. While Lorentzian signature is guaranteed from the structure group $H$ (Lorentz group) and the metric $\eta$, we shall see in the next section when we derive the equations of motion that the sign of $kk'$ determines if we obtain a positive or negative cosmological field $\Lambda_G$.
	
	Let $\omega \in \Omega^q(\mathcal{M},\g) $ be a q-form on $\mathcal{M}$ with values in $\g$.
	Its local trivialisation in a given chart is :
	\begin{align}
		\omega= \dfrac{1}{r!} \omega_{\mu_1\mu_2...\mu_r}dx^{\mu_1}\wedge dx^{\mu_2}\wedge ... \wedge dx^{\mu_r}
	\end{align}
	From the metric $g$ we define the Hodge star operator $*$ that acts on $\omega$ as :
	\begin{align}
		*\omega = & \dfrac{1}{r!} \sqrt{|det(g)|}\,\omega_{\mu_1...\mu_r} 
		g^{\mu_1 \nu_1}...g^{\mu_r \nu_r} \varepsilon_{\nu_1...\nu_m} dx^{\nu_{r+1}}\wedge ... dx^{\nu_m}
	\end{align}
	Similarly let $\tilde{*}$ be the Hodge star operator associated to $\tg$. We also introduce a nondegenerate sesquilinear form on $\mathbb{C}^{4}$ : $(\psi,\chi) = \psi^\dagger \gamma^0 \chi = \widebar{\psi} \chi$ which is used to define a symmetric metric on $\mathbb{C}^{4}$ as $h^\varepsilon(\psi,\chi) = \dfrac{1}{2} \bigl((\psi,\chi) + (\chi,\psi) \bigl)$. 
	
	The most conservative choice of matter action corresponds to minimal coupling\footnote{In the sense that the spacetime dependent parameters $k,k',\ell$ do not enter the matter action, leaving only a minimal coupling to $A$ and $\beta$.}:
	\begin{align}
		\label{Matter action 3}
		\hskip -7mm
		S_M[A,\beta] & = \int_{\mathcal{M}} h^\varepsilon\bigl( i \psi, h_{ab} \gamma^a \varpim^b \wedge *D_\h\psi\bigl) 
		=\int_{\mathcal{M}} \dfrac{36(kk')^2\kappa^2\zeta}{\ell^4} h^\varepsilon\bigl( i \psi, \eta_{ab} \gamma^a \beta^b \wedge \tilde{*}D_\h\psi\bigl)
		= \mathrm{Re} S_D
	\end{align}
	where we set $\kappa^2 \zeta = \tfrac{\ell^4}{216(kk')^2}$ 
	such that the overall factor is arranged to stick with the Dirac action given in \cite[p.197]{gockeler_differential_1987} and with $\mathrm{Re} S_D  = \mathrm{Re}\int_{\mathcal{M}} \mathcal{L}_D$ corresponding to the real part of the Dirac action.
	
	However what we deem to be the most natural matter action (in our setting) modeled after the Dirac action corresponds to:
	\begin{align}
		\label{Matter action l}
		S'_M[A,\varpim] & = \int_{\mathcal{M}} h^\varepsilon\bigl( i \psi, h_{ab} \gamma^a \varpim^b \wedge *D_\h\psi\bigl) = \mathrm{Re}\int_{\mathcal{M}} \dfrac{(\tk \tk')^2}{\tl^4} \mathcal{L}_D
	\end{align}
	where this time we set $\kappa^2 \zeta = \tfrac{\ell_0^4}{216(k_0k_0')^2}$ with $k_0=k(x_0)$ and $\tilde{k} = \dfrac{k}{k_0}$, $x_0 \in \mathcal{M}$ (the same goes for $k'$ and $\ell$). This way $S'_M$ coincides with $\mathrm{Re}S_D$ at a reference point $x_0$, and takes into account the degrees of freedom hidden in $k$, $k'$ and $\ell$.
	
	One may also opt for other non-trivial matter actions expressed as:
    \begin{align}
        \label{f matter action}
        S^f_M[A,\varpim] = f(k,k',\ell) \mathrm{Re}S_D
    \end{align}
    where $f(k,k',\ell)$ is a function of the mutation scalar fields.
    
    To make link with physical quantities like the cosmological constant and gravitational coupling we shall now derive the equations of motion associated to the total gauge + matter action.
	\vspace{-0.4cm}
	\subsubsection*{Equations of motion}
	\label{Lorentzian Equations of motion}
	
	\vspace{-0.1cm}
	Computing the equations of motion (EOM's) relative to $\beta$, $\varpih=A$, $\ell$, $k$ and $k'$ with action $S_T [\beta, A , \ell, k, k'{]}= S_G[\beta, A , \ell, k ,k'{]} + S_M[\beta, A{]}$ or $S'_T [\beta, A , \ell, k, k']= S_G[\beta, A , \ell, k ,k'] + S'_M[\beta, A, \ell, k, k']$ gives :\vspace{-0.4cm}
	{\small
	\begin{align}
		\dfrac{\delta \mathcal{L}_G [A,\beta,\ell,k,k']}{\delta \beta^c} &= 
		%%%%%
		\dfrac{ekk'}{8\pi^2\ell^2} ( R^{ab} + \dfrac{ kk' }{\ell^2} \beta^a \wedge \beta^b ) \wedge \beta^d \varepsilon_{abcd} 
		+ \dfrac{ykk'}{2\pi^2\ell^2} R^{ab} \wedge \beta_b \delta_{ac}
		%%%%%%%%
		= \left\{
		\begin{array}{ll}
			\tau_c \text{ for } S_M \\
			\dfrac{(\tk \tk')^2}{\tl^4} \tau_c \text{ for } S'_M
		\end{array}	
		\right.
		\label{eom frame Lorentz quasi} \\
		%%%%%%%%%%%%%%%%%%%%%%%%%%%%%%%%%%%%%%%%%%%%%%%%%%%%%%%%%
		\dfrac{\delta \mathcal{L}_G [A,\beta,\ell,k,k']}{\delta A^{ab}}
		%%%%%%%%
		&=
		\dfrac{ekk'}{8\pi^2\ell^2} \bigl(T^c + \dfrac{1}{2} (\dfrac{dk}{k} + \dfrac{dk'}{k'} - 2 \dfrac{d\ell}{\ell}) \wedge \beta^c \bigr) \wedge \beta^d \varepsilon_{abcd} \nonumber \\
		& \hspace{-0.3cm}+ \dfrac{ykk'}{2\pi^2\ell^2} \bigl( T_a + \dfrac{1}{2} (\dfrac{dk}{k} + \dfrac{dk'}{k'} - 2 \dfrac{d\ell}{\ell}) \wedge \beta_a \bigr)  \wedge \beta_b
		%%%%%%%%%
		= \left\{
		\begin{array}{ll}
			\dfrac{1}{2} \mathfrak{s}_{ab} \text{ for } S_M \\
			\dfrac{(\tk \tk')^2}{2\tl^4} \mathfrak{s}_{ab} \text{ for } S'_M
		\end{array}	
		\right.
		\label{eom spin Lorentz quasi} 
	\end{align}
}%
\vspace{-0.0cm}
	The EOM's relative to $\ell$, $k$ and $k'$ are respectively : 
	{\small
	\begin{align}
		\dfrac{\delta \mathcal{L}_G }{\delta \ell}  = &- \dfrac{2kk'}{\ell^3} \biggl(
		\dfrac{e}{16\pi^2 } ( R^{ab} + \dfrac{kk'}{\ell^2} \beta^a \wedge \beta^b ) \wedge \beta^c \wedge \beta^d \varepsilon_{abcd} 
		%%%
		+ \dfrac{y}{4 \pi^2}
		 R^{ab} \wedge \beta_a \wedge \beta_b\nonumber
		\biggr) \\
%%%%%%%%%%%%%%%%%%%
		= & \left\{
		\begin{array}{ll}
			0 \text{ for } S_M \\
			\dfrac{ 4(\tk \tk')^2}{l_0\tl^5} h^\varepsilon( i \psi, \eta_{ab} \gamma^a \beta^b \wedge *D_\h\psi) \text{ for } S'_M
		\end{array}	
		\right.
		\label{eom l Lorentz}  \\
		%%%%%%%%%%%%%%%%%%%%%%%%%%%%%%%%%%%%%%%%%%%%%%%%%%%%%%%%%%
		\dfrac{\delta \mathcal{L}_G }{\delta k} 
		= &\dfrac{k'}{\ell^2} \biggl(
		\dfrac{e}{16\pi^2 } ( R^{ab} + \dfrac{kk'}{\ell^2} \beta^a \wedge \beta^b ) \wedge \beta^c \wedge \beta^d \varepsilon_{abcd}
		%%%%%%%%%
		+ \dfrac{y}{4 \pi^2}
		 R^{ab} \wedge \beta_a \wedge \beta_b\nonumber
		%%%%%%%%%
		\biggr) \\
		%%%%%%%%%
		= & \left\{
		\begin{array}{ll}
			0 \text{ for } S_M \\
			- \dfrac{2\tk \tk'^2 }{k_0\tl^4} h^\varepsilon( i \psi, \eta_{ab} \gamma^a \beta^b \wedge *D_\h\psi) \text{ for } S'_M
		\end{array}	
		\right.
		\label{eom k Lorentz} \\
		%%%%%%%%%%%%%%%%%%%%%%%%%%%%%%%%%%%%%%%%%%%%%%%%%%%%%%%%%%
		\dfrac{\delta \mathcal{L}_G }{\delta k'} 
		= &\dfrac{k}{\ell^2} \biggl(
		\dfrac{e}{16\pi^2 } ( R^{ab} + \dfrac{kk'}{\ell^2} \beta^a \wedge \beta^b ) \wedge \beta^c \wedge \beta^d \varepsilon_{abcd}
		%%%%%%%%%
		+ \dfrac{y}{4 \pi^2}
		R^{ab} \wedge \beta_a \wedge \beta_b\nonumber
		%%%%%%%%%
		\biggr) \\
		%%%%%%%%%
		= & \left\{
		\begin{array}{ll}
			0 \text{ for } S_M \\
			- \dfrac{2\tk^2 \tk' }{k'_0\tl^4} h^\varepsilon( i \psi, \eta_{ab} \gamma^a \beta^b \wedge *D_\h\psi) \text{ for } S'_M
		\end{array}	
		\right.
		\label{eom k' Lorentz}
	\end{align}
}%
	We will begin by studying the case $S=S_T$ corresponding to a \textit{minimal coupling} to matter. 
	Unless $k=0$ or $k'=0$ the last $3$ equations imply the same constraint, that is:
	\begin{align}
		\dfrac{e}{16\pi^2 } ( R^{ab} + \dfrac{kk'}{\ell^2} \beta^a \wedge \beta^b ) \wedge \beta^c \wedge \beta^d \varepsilon_{abcd}
		%%%%%%%%%
		+  \dfrac{y}{4 \pi^2}
		 R^{ab} \wedge \beta_a \wedge \beta_b
		%%%%%%%%%
		%
		%%%%%%%%%
		= 0
	\end{align}
	
	Applying the Hodge star to this equation yields: 
	\begin{align}
		\dfrac{e}{8\pi^2 } ( \overbrace{\rR}^{=\Rmod[ab]{ab}} + \dfrac{12kk'}{\ell^2} )
		%%%%%%%%%
		+ \dfrac{y}{8 \pi^2}
		 \Rmod[\mu\nu]{ab} \varepsilon^{\mu\nu}_{~~ab}
		%
		%%%%%%%%%
		= 0
	\end{align}
	
	Finally applying the Hodge star on \ref{eom frame Lorentz quasi} gives us:
	\begin{align}
		& G_{kc} 
		- \dfrac{3kk'}{\ell^2} \eta_{ck}
		+ \dfrac{y}{e} \Rmod[rs]{ab} \varepsilon^{rs}_{~~kb} \delta_{ac}
		%%%%
		= - \dfrac{4\pi^2 \ell^2}{ekk'} \tau_{ck} \label{eom frame lorentz 2}
	\end{align}
	From \eqref{eom frame lorentz 2} we identify the gauge contribution to the cosmological field $\Lambda_G = - \dfrac{3kk'}{\ell^2}$ and the gravitational coupling field $G = - \dfrac{\pi \ell^2}{2ekk'} $.
	Therefore, for $k,k'\neq0$, the equations of motion on $\beta$, $A$, $\ell$, $k$ and $k'$ are equivalent to:
	\begin{align}
		G_{kc} 
		+ \dfrac{\chi}{4} \eta_{ck}
		+ \dfrac{y}{e} \Rmod[rs]{ab} \varepsilon^{rs}_{~~kb} \delta_{ac}
		%%%%
		= \dfrac{48\pi^2}{e\chi} \tau_{ck} \label{eom frame lorentz 3} \\
		%%%%%%%%%%%%%%%%%%%%%%%%%%%%%%%%%%%%%%%%%%%%%%%%%%%%%%%%%
		\hspace{-1cm}\bigl(T^c + \dfrac{1}{2} \dfrac{d \chi}{\chi} \wedge \beta^c \bigr) \wedge \beta^d \varepsilon_{abcd}
		+ \dfrac{4y}{e} \bigl( T_a + \dfrac{1}{2} \dfrac{d \chi}{\chi} \wedge \beta_a \bigr)  \wedge \beta_b
		%%%%%%%%%
		= - \dfrac{48\pi^2}{e\chi} \mathfrak{s}_{ab} \label{eom spin Lorentz 3} \\
		%%%%%%%%%%%%%%%%%%%%%%%%%%%%%%%%%%%%%%%%%%%%%%%%%%%%%%%%%
		\chi = - \dfrac{12 kk'}
		{\ell^2} = \rR + \dfrac{y}{e} \Rmod[\mu\nu]{ab} \varepsilon^{\mu\nu}_{~~ab} \label{eom csm cst}
	\end{align}
	Interestingly, if one allows for a degenerate metric (an example being a null $k$ or $k'$ which would correspond to Poincaré moving geometry) at specific spacetime points $x_i$, it may in principle be possible to transition from a positive to a negative $\Lambda_G$ and vice-versa for vacuum solutions ($\tau_{ck}(x_i) = \mathfrak{s}_{ab} (x_i) = 0$) as well as other less trivial solutions with a localized null energy-momentum tensor $\tau_{ck}(x_i) =0$.
	This can be deduced by looking at the first non-simplified expressions \eqref{eom frame Lorentz quasi} \eqref{eom spin Lorentz quasi} of the equations of motion. In particular, applying the $\tilde{*}$ Hodge to \eqref{eom frame Lorentz quasi} leads to the following equations (valid for points such that $\chi(x_i) = 0$):
	\begin{align}
		\chi G_{kc} + \dfrac{\chi^2}{4} \eta_{ck} + \chi \dfrac{y}{e} {\Rmod[rs]{ab}} \varepsilon^{rs}_{~~kb} \delta_{ac} =\left\{
		\begin{array}{ll}
			\dfrac{48\pi^2}{e} \tau_{ck} \text{ for } S_M \\
			\dfrac{48\pi^2}{e} \tilde{\chi}^2 \mathfrak{\tau}_{ck} \text{ for } S'_M
		\end{array}	
		\right. \label{eom non deg}
	\end{align}

	For null $k(x_i)$ or $k'(x_i)$, \eqref{eom frame Lorentz quasi} and \eqref{eom spin Lorentz quasi} become:
	\begin{align}
		\tau_c & = 0 \text{ for } S_M \label{frame general}\\
		( \dfrac{e}{8\pi^2\ell^2} \beta^c \wedge \beta^d \varepsilon_{abcd} 
		+ \dfrac{y}{2\pi^2\ell^2} \beta_a \wedge \beta_b ) \wedge (k'dk + kdk')
		%%%%%%%%%
		& = \left\{
		\begin{array}{ll}
			\mathfrak{s}_{ab} \text{ for } S_M \\
			0 \text{ for } S'_M
		\end{array}	
		\right. \label{spin general}
	\end{align}
	If both $k(x_i) = k'(x_i)=0$ then this corresponds to vacuum $\tau_c(x_i)=\mathfrak{s}_{ab}(x_i)=0$ for $S_M$ and only $\tau_c(x_i)=0$ for $S'_M$. Therefore matter with non null energy momentum tensor $\tau_c(x_i)\neq 0$ at point $x_i$ is incompatible with $\Lambda_G(x_i)=0$ from the equation of motion \eqref{eom frame Lorentz quasi}.
	
	Another more involved case to look for is $k(x_i)=0$ with $dk(x_i)\neq 0$ and $k'(x_i)\neq0$ or vice versa. This configuration leads to:
	\begin{align}
		\tau_c & = 0 \text{ for } S_M \label{ghost spinor 1} \\
		( \dfrac{e}{8\pi^2\ell^2} \beta^c \wedge \beta^d \varepsilon_{abcd} 
		+ \dfrac{y}{2\pi^2\ell^2} \beta_a \wedge \beta_b ) k'dk
		%%%%%%%%%
		& = \left\{
		\begin{array}{ll}
			\mathfrak{s}_{ab} \text{ for } S_M \\
			0 \text{ for } S'_M
		\end{array}	
		\right. \label{ghost spinor 2}
	\end{align}
	An important thing to note for $S_M$ is that equations \eqref{frame general} and \eqref{spin general} allow for ghost spinors \cite{griffiths1979dirac,boehmer2008dark} characterized by a null energy momentum tensor $\tau_c=0$ and non null spin-density $\mathfrak{s}_{ab}\neq 0$ depending on the values of the paremeters $e,y$, the tetrad and $k,k'$.
	
	On the other hand for the matter action $S'_M$ the equation of motion \eqref{eom frame Lorentz quasi} is trivially satisfied for null $k(x_i)$ or $k'(x_i)$ since both sides vanish whereas for null $k(x_i)$ or $k'(x_i)$ \eqref{eom spin Lorentz quasi} is satisfied if either \begin{align}
		k'dk = -kdk' \text{ or }  \dfrac{e}{8\pi^2\ell^2} \beta^c \wedge \beta^d \varepsilon_{abcd} 
		+ \dfrac{y}{2\pi^2\ell^2} \beta_a \wedge \beta_b = 0
	\end{align}
	One can also apply the covariant derivative $D$ %(\cdot) = d \cdot + A \wedge \cdot$
	defined from the spin connection to equation \eqref{eom non deg}, which yields:
	\begin{align}
		D^k \chi ( G_{kc} + \dfrac{\chi}{2} \eta_{ck} + \dfrac{y}{e} {\Rmod[rs]{ab}} \varepsilon^{rs}_{~~kb} \delta_{ac} ) + \chi D^k (G_{kc} +  \dfrac{y}{e} {\Rmod[rs]{ab}} \varepsilon^{rs}_{~~kb} \delta_{ac}) = 
		\dfrac{48\pi^2}{e} D^k \tau_{ck}  \label{consistency condition}
	\end{align}
	In case $\chi(x_i)=0$ it simplifies at $x_i$ to:
	\begin{align}
		D^k \chi ( G_{kc} + \dfrac{y}{e} {\Rmod[rs]{ab}} \varepsilon^{rs}_{~~kb} \delta_{ac} )= 
	\dfrac{48\pi^2}{e} D^k \tau_{ck},
	\end{align}
	thus offering more information on the on-shell variations of $\chi$ at such degenerate configurations.
	
	The examination of \eqref{eom non deg} in search of perturbative dS/AdS crossing solutions and their analysis shall be the subject of future work.
	\begin{remark}
		It is interesting to note that the condition \eqref{condition stay algebra} for the curvature to stay valued in the Lie algebra $\g$ doesn't restrict the spin-density to be null in equation \eqref{spin general} and simplifies \eqref{spin general} to:
		\begin{align}
			( \dfrac{e}{4\pi^2\ell^2} \beta^c \wedge \beta^d \varepsilon_{abcd} 
			+ \dfrac{y}{\pi^2\ell^2} \beta_a \wedge \beta_b ) k'dk
			%%%%%%%%%
			& = \left\{
			\begin{array}{ll}
				\mathfrak{s}_{ab} \text{ for } S_M \\
				0 \text{ for } S'_M
			\end{array}	
			\right. \label{spin general 2}
		\end{align}
		However, connections verifying \eqref{condition stay algebra} at every point do not allow analytic transitions from positive to negative $\Lambda_G$ or vice versa since in this case $k' \propto k$ and $\Lambda_G \propto (\dfrac{k}{\ell})^2$ doesn't allow a sign change.
	\end{remark}

	Further inspection of \eqref{eom spin Lorentz 3} shows it links torsion, $\chi$ and the spin density of matter according to the differential equation:
	\begin{align}
		\label{torsion diff eq}
		T^c_{~ab} \varepsilon^{ab}_{~~cr} + \dfrac{y}{e} ( 8 \Tmod[rc]{c} +12 \dfrac{\partial_r\chi}{\chi} )  = - \dfrac{48\pi^2}{e\chi} \mathfrak{s}^{ab}_{~~abr}
	\end{align}
One may see that this equation is a generalization of the equation linking axial torsion $T^c_{~ab} \varepsilon^{ab}_{~~cr}$ to the spin density $\mathfrak{s}$ of matter fields in Einstein-Cartan gravity theories. Major differences arise when $y\neq 0$ (Holst action) such that trace torsion $\Tmod[rc]{c}$ and variations of $\chi$ are now also sourced by $\mathfrak{s}$.
One may then use the stringent bounds on torsion (especially trace torsion) described in \cite{kostelecky2008constraints} to infer that even a small variation $\dfrac{\partial_r\chi}{\chi} >> 10^{-26}$ GeV would overshadow in equation \eqref{torsion diff eq} the maximal trace torsion $\Tmod[rc]{c}$ contribution allowed by current experimental bounds.

\begin{remark}
	Although \eqref{torsion diff eq} is a differential equation for $\chi$ and the scalar fields $k,k',\ell$, the absence of a quadratic kinetic term of the form $d\chi \wedge * d\chi$ in the action makes the current formulation differ from dynamical Chern-Simons \cite{alexander_chern-simons_2009} and Brans-Dick theories \cite{de2010generalized}. We leave the introduction of such $d\chi \wedge * d\chi$ term and the study of the resulting dynamical degrees of freedom to future work. One can also see via relation \eqref{torsion diff eq} that although both axial torsion $T^c_{~ab} \varepsilon^{ab}_{~~cr}$ and trace torsion $\Tmod[rc]{c}$  vary with $\chi$ and $\mathfrak{s}$ they do not propagate.
\end{remark}

	As is for example mentionned in \cite{sola_vacuum_2014}, vacuum energy density $\rho_{\text{vac}}$ is inter alia subject to variations when there is symmetry breaking, therefore, we will assume it depends on the point $x$ of the manifold $\mathcal{M}$.
	Separating the energy-momentum tensor $\tau_{ck}$ as $\tau_{ck}=\tau_{M,ck} -  \rho_{\text{vac}} \eta_{ck}$ with $\tau_{\text{vac},ck}=-\rho_{\text{vac}}\eta_{ck}$ \cite{carroll_cosmological_2001} the part related to the vacuum energy density $\rho_{\text{vac}}$ leads to:
	
	\begin{align}
		G_{kc} 
		+ \Lambda \eta_{ck}
		- \dfrac{y}{e} \Rmod[rs]{ab} \varepsilon^{rs}_{~~bk} \delta_{ac}
		%%%%
		&= 8\pi G \tau_{M,ck} \label{eom frame lorentz 4}
	\end{align}
	with:
	\begin{align*}
		\Lambda = &\Lambda_G + \Lambda_\text{vac} = \Lambda_G + 8 \pi G \rho_{\text{vac}} = \overbrace{\dfrac{\chi}{4}}^{\Lambda_G} + \overbrace{\dfrac{48\pi^2\rho_{\text{vac}}}{e\chi}}^{\Lambda_\text{vac}} \numberthis \\[2mm] =&\footnotemark \dfrac{1}{4} (\rR+ \dfrac{y}{e} \Rmod[\mu \nu]{ab} \varepsilon^{\mu \nu}_{~~ab} ) 
		+ \dfrac{48 \pi^2 \rho_{\text{vac}}}{e \rR + y \Rmod[\mu\nu]{ab} \varepsilon^{\mu\nu}_{~~ab}} 
	\end{align*}
	\footnotetext{Where \eqref{eom csm cst} has been used.} where $\Lambda_G$ and $\Lambda_{\text{vac}}$ are respectively the gauge action and vacuum contribution to the cosmological field and $G = \dfrac{6 \pi }{e\chi}$ is the varying gravitational coupling which coincides with Newton's constant $G(x_0)=G_0$ at a point $x_0 \in \mathcal{M}$. Similarily we will denote $\Lambda_G(x_0) = \Lambda_0$ the bare cosmological constant at point $x_0$ and $\rR(x_0) = \Rmod[ab]{ab} (x_0)= \rR_0$; $\Rmod[\mu \nu]{ab} (x_0) \varepsilon^{\mu \nu}_{~~ab} = \Rzmod[\mu \nu]{ab} \varepsilon^{\mu \nu}_{ab}$ as well as $\rho_{\text{vac}}(x_0) = \rho_{\text{vac}}^0$, $\chi(x_0)=\chi_0$.
	
	The first thing to notice is that $\Lambda_G = \dfrac{3\pi}{2eG}$. It is also possible to interpret the variation of $G$ as a source of dark matter. If one assumes $\rho_{vac}=0$ this leads to a model where dark energy ($\Lambda_G$) is inversely proportional to the gravitational coupling $G$ (which may act as a source of dark matter by changing value depending on the region of spacetime).
	
	With these notations the constants $e$, $r$ and $y$ appearing in the linear combination of invariant polynomials in the action are entirely determined by these three relations:
	\begin{align}
		\Lambda_0 &= \dfrac{\chi_0}{4} \\
		G_0 &= \dfrac{6 \pi }{e\chi_0} \\
		\gamma &
		= \dfrac{e}{r+2y} 
	\end{align}
	where $\gamma$ is the Barbero-Immirzi parameter.
	In the end we can express them as:
	\begin{align}
		e &= \dfrac{\chi_0}{4\Lambda_0} = \dfrac{3\pi}{2\Lambda_0G_0} \\
		r &= \left. 
		\dfrac{3 \pi ( \rR + \tfrac{1}{2\gamma} \Rmod[\mu \nu]{ab} \varepsilon^{\mu \nu}_{~~ab} - \Lambda_0) }{\Lambda_0 G_0 \Rmod[\mu \nu]{ab} \varepsilon^{\mu \nu}_{~~ab}  }
		\right\rvert_{x_0} \\
		y &= \dfrac{e}{2\gamma} - \dfrac{r}{2}
	\end{align}
	An important thing to note is that $\Lambda_0$  is the bare cosmological constant at a point $x_0 \in \mathcal{M}$, not the effective one. If one knows the vacuum energy density $\rho_{\text{vac}}^0$ as well as the effective cosmological constant $\Lambda (x_0)$ and Newton's constant $G_0$ at a point $x_0$ of spacetime, it can be deduced from $\Lambda_0 = \Lambda (x_0) - 8 \pi G_0 \rho_{\text{vac}}^0$.
	
	In case one doesn't know $\rho_{\text{vac}}^0$ these relations may allow to predict its value by trying to measure $\Lambda_0$ via its effects in the different terms of the action.
	
	One can also allow for $e$, $r$ and $y$ to vary depending on the points of the manifold but in this case we in general, lose the topological nature of the action.
	
	From the expression of $\Lambda_G = - \dfrac{3kk'}{\ell^2}$ and $G = - \dfrac{\pi \ell^2}{2ekk'} $ we derive the variations:
	\begin{align}
		d\Lambda_G &= \dfrac{3kk'}{\ell^2} ( 2 \dfrac{d\ell}{\ell} - \dfrac{dk}{k} - \dfrac{dk'}{k'} ) \\
		dG &= \dfrac{\pi \ell^2}{2ekk'} ( \dfrac{dk}{k} + \dfrac{dk'}{k'} - 2 \dfrac{d\ell}{\ell})
	\end{align}
    As in \cite{thibaut2024gravity}, the action becomes topological in the limit where $\Lambda_G \rightarrow 0$. This makes $\Lambda_G$ acts as a regulator for the theory in this kind of limit and generalizes the regulator interpretation of the cosmological constant of \cite{aros_conserved_2000} to the spacetime-dependent dark energy described by $\Lambda_G$ in our model.
	Thus, at least in this theory, the topological property of the gauge action is intimately tied to observing a small value of $\Lambda_G$ but does not require it to be a constant. In case $\Lambda_{\text{vac}}=0$ a small $\Lambda_G$ is in fact quite coherent with the experimental value of the cosmological constant measured in the framework of the standard model of cosmology.

	In the simplified case corresponding to $y=0$, with \textit{minimal coupling} ($S_T = S_G + S_M$) coupled with the assumption that $R$ decreases with cosmological time leads to a gauge contribution to dark energy $\Lambda_G = \dfrac{\rR}{4}$ of the Ricci type (see e.g. \cite{gao_holographic_2009})
	with the particularity that it decreases over time in favor of a growing gravitational coupling to matter $G= \dfrac{6 \pi }{e\rR}$.
	This growing $G$ may be interpreted as a source of dark matter since considering a model with a constant $G$ like the $\Lambda$-CDM model would then lead to underestimating the contribution of matter to curvature in the Einstein equations at late cosmological times.
	Such a transfer from dark energy to dark matter in fact corresponds to scenarios which could alleviate the Hubble tension as discussed in \cite{tang_uniting_2025,collaboration_desi_2025}. 
	
	Additionally, the fact that $\Lambda_G \propto \rR$ and $G \propto \dfrac{1}{\rR}$ leads to some kind of Lenz law. In the limit where the Ricci curvature $\rR$ diverges, the gravitional coupling $G$ tends to $0$, thus decoupling curvature from matter while dark energy $\Lambda_G$ diverges, favoring the dispersion of matter. Such a behavior may allow to evade some of the curvature singularities (see for example \cite{bonanno2024dust} for similar results) encountered in solutions of the usual Einstein equations. This shall also be examined more thoroughly from the point of view of asymptotically safe gravity (see for example \cite{eichhorn2022black}) and the functional renormalization group \cite{dupuis2021nonperturbative}, since the mutation \eqref{mutation Lorentz} acts as a change of scale within the Lie algebras/groups. 
	
	The fact that $G \propto \dfrac{1}{\rR}$ also means that in the outermost region of galaxies (with less curvature), $G$ should then have a higher value. Although different, this is in fact reminiscent of the behavior of MOdified Newtonian Dynamics (MOND) \cite{famaey2012modified}, MOdified Gravity(MOG)/Scalar-tensor-vector gravity \cite{moffat2006scalar} and other similar modified theories of gravity. In the future it would be interesting to investigate this behavior to see if it could help describe the rotation curve of galaxies as was done for example for various models in \cite{chae2022distinguishing,parbin2023galactic,shabani2023galaxy,mohan2024galactic}.
	
	This toy model with $r=y=0$ basically corresponds to MacDowell Mansouri gravity with spacetime-dependent fields $k,k',\ell$. For such a situation EOM's \eqref{eom frame lorentz 3}-\eqref{eom csm cst} then become:
	\begin{align}
		G_{kc} 
		- \dfrac{3kk'}{\ell^2} \eta_{ck}
		%%%%
		&= - \dfrac{4\pi^2 \ell^2}{ekk'} \tau_{ck} \label{eom frame Lorentz 4} \\
		%%%%%%%%%%%%%%%%%%%%%%%%%%%%%%%%%%%%%%%%%%%%%%%%%%%%%%%%%
		\bigl(T^c + \dfrac{1}{2} (\dfrac{dk}{k} + \dfrac{dk'}{k'} - 2 \dfrac{d\ell}{\ell}) \wedge \beta^c \bigr) \wedge \beta^d \varepsilon_{abcd}
		%%%%%%%%%
		&= \dfrac{4\pi^2 \ell^2}{ekk'} \mathfrak{s}_{ab} \label{eom spin Lorentz 4} \\
		%%%%%%%%%%%%%%%%%%%%%%%%%%%%%%%%%%%%%%%%%%%%%%%%%%%%%%%%%
		\ell^2 &= - \dfrac{12 kk'
		}
		{\rR} \label{eom l 4}
	\end{align}
	
	Using \eqref{eom l 4} in \eqref{eom frame Lorentz 4} and \eqref{eom spin Lorentz 4} yields:
	
	\begin{align}
		G_{kc} 
		+ \dfrac{\rR}{4} \eta_{ck}
		%%%%
		&= \dfrac{48\pi^2}{e\rR} \tau_{ck} \label{eom frame Lorentz 5} \\
		%%%%%%%%%%%%%%%%%%%%%%%%%%%%%%%%%%%%%%%%%%%%%%%%%%%%%%%%%
		\bigl(T^c + \dfrac{1}{2} \dfrac{d\rR}{\rR} \wedge \beta^c \bigr) \wedge \beta^d \varepsilon_{abcd}
		%%%%%%%%%
		&= - \dfrac{48\pi^2}{e\rR} \mathfrak{s}_{ab} \label{eom spin Lorentz 5} \\
		%%%%%%%%%%%%%%%%%%%%%%%%%%%%%%%%%%%%%%%%%%%%%%%%%%%%%%%%%
		\rR &= - \dfrac{12 kk'
		}
		{\ell^2} \label{eom l 5}
	\end{align}
	
	Therefore, in the event we fix $r=y=0$, the effective cosmological constant is $\Lambda = \Lambda_G + \Lambda_{\text{vac}} = \dfrac{\rR}{4} + \dfrac{48 \pi^2 \rho_{\text{vac}}}{e\rR}$ with $e= \dfrac{6 \pi}{G_0 \rR_0}$ and the gravitational coupling becomes $G = \dfrac{G_0}{\tilde{\rR}}$
	
	If we further impose $\rho_{\text{vac}}=0$ we then obtain $\Lambda = 8 \pi G_0 \rho_{\Lambda} = \dfrac{\rR}{4}$ with $\rho_{\Lambda} = \dfrac{\alpha}{16 \pi G_0} \rR$. As indicated previously $\Lambda$ then corresponds to a Ricci dark energy model like the one described in \cite{gao_holographic_2009}, with parameter $\alpha = \dfrac{1}{2}$ being quite close to the estimated best fit value $\alpha_{e} \simeq 0.46$ of the previously mentionned paper. This specific case also appears to be compatible with what has been derived in \cite{sola_vacuum_2014} and corresponds to the .
	However we argue that imposing $r=y=\rho_{\text{vac}}=0$ is too restrictive.
	
	Let $\tau$ be the trace of the energy-momentum tensor.
	For a general $y$, we now take the trace of \eqref{eom frame lorentz 3} and combine it with \eqref{eom csm cst}:
	\begin{align}
		y \Rmod[rs]{ab} \varepsilon^{rs}_{~~ab} = 4\pi e G \tau
	\end{align} 
	Considering the case of a perfect fluid ($w = \dfrac{P}{\rho}$) leads to:
	\begin{align}
		y \Rmod[rs]{ab} \varepsilon^{rs}_{~~ab} = 4\pi e G (\rho - 3 P)
	\end{align} 
	and allows to obtain a novel interpretation of $y$ as a coupling between non radiative matter ($w \neq 1/3$) and pseudo-scalar curvature $\Rmod[rs]{ab} \varepsilon^{rs}_{~~ab}$. More specifically the case $y=0$ only allows the vacuum ($\tau_{ck} = 0$) or radiative matter ($w = 1/3$) on-shell and confirms the need to consider $y\neq 0$.
	\medskip
	
	In case we use the total action $S'_T = S_G + S'_M$, applying similar operations on the system of EOM's \eqref{eom frame Lorentz quasi}-\eqref{eom k' Lorentz} as done above instead yields:
	\begin{align}
		G_{kc} 
		- \dfrac{3kk'}{\ell^2} \eta_{ck}
		+ \dfrac{y}{e} \Rmod[rs]{ab} \varepsilon^{rs}_{~~bk} \delta_{ac}
		%%%%
		= - \dfrac{4\pi^2 \ell_0^2}{ek_0k_0'} \dfrac{\tk \tk'}{\tl^2} \tau_{ck} \\
		%%%%%%%%%%%%%%%%%%%%%%%%%%%%%%%%%%%%%%%%%%%%%%%%%%%%%%%%%
		\bigl(T^c + \dfrac{1}{2} (\dfrac{dk}{k} + \dfrac{dk'}{k'} - 2 \dfrac{d\ell}{\ell}) \wedge \beta^c \bigr) \wedge \beta^d \varepsilon_{abcd}
		+ \dfrac{4y}{e} \bigl( T_a + \dfrac{1}{2} (\dfrac{dk}{k} + \dfrac{dk'}{k'} - 2 \dfrac{d\ell}{\ell}) \wedge \beta_a \bigr)  \wedge \beta_b \nonumber \\
		%%%%%%%%%
		= \dfrac{4\pi^2\ell_0^2}{ek_0k_0'} \dfrac{\tk \tk'}{\tl^2} \mathfrak{s}_{ab} \\
		%%%%%%%%%%%%%%%%%%%%%%%%%%%%%%%%%%%%%%%%%%%%%%%%%%%%%%%%%
		\dfrac{e}{8\pi^2 } (\rR + \dfrac{12kk'}{\ell^2} )
		%%%%%%%%%
		+ \dfrac{y}{8 \pi^2}
		\Rmod[\mu\nu]{ab} \varepsilon^{\mu\nu}_{~~ab}
		%%%%%%%%%
		= \dfrac{3\ell_0^2}{k_0k_0'} \dfrac{\tk \tk'}{\tl^2} \sqrt{|\tg|} h^\varepsilon( i \psi, \eta_{ab} \gamma^a \beta^b_\mu D^\mu_\h\psi) \label{quad l eq 2}
	\end{align}
	
	By identifying $G_0 = \dfrac{3\pi}{2\Lambda_0 e}$ and $\left.\chi'_0 = \dfrac{1}{2} \Bigl(\chi_0 \pm \sqrt{\chi_0^2- 1152\pi^2/e \sqrt{|\tg|} h^\varepsilon( i \psi, \eta_{ab} \gamma^a \beta^b_\mu D^\mu_\h\psi)} \Bigr) \right\rvert_{x_0}
	 $ we finally obtain:
	\begin{align}
		G_{kc} 
		+ \dfrac{\chi'}{4} \eta_{ck}
		- \dfrac{y}{e} \Rmod[rs]{ab} \varepsilon^{rs}_{~~bk} \delta_{ac}
		%%%%
		= 8\pi G_0 \tilde{\chi}' \tau_{ck} \label{eom frame lorentz '} \\
		%%%%%%%%%%%%%%%%%%%%%%%%%%%%%%%%%%%%%%%%%%%%%%%%%%%%%%%%%
		\hspace{-1cm}\bigl(T^c + \dfrac{1}{2} \dfrac{d \chi'}{\chi'} \bigr) \dfrac{\ell^2}{kk'} \wedge \beta^d \varepsilon_{abcd}
		+ \dfrac{4y}{e} \bigl( T_a + \dfrac{1}{2} \dfrac{d \chi'}{\chi'} \wedge \beta_a \bigr)  \wedge \beta_b
		%%%%%%%%%
		= - 8\pi G_0 \tilde{\chi}' \mathfrak{s}_{ab} \label{eom spin Lorentz '} \\
		%%%%%%%%%%%%%%%%%%%%%%%%%%%%%%%%%%%%%%%%%%%%%%%%%%%%%%%%%
		\chi' = - \dfrac{12 kk'}
		{\ell^2} = \overbrace{\bigl(\rR + \dfrac{y}{e} \Rmod[\mu\nu]{ab} \varepsilon^{\mu\nu}_{~~ab}\bigr)}^\chi / \bigl( 1 - \dfrac{288 \pi^2}{e\chi_0'^2} \sqrt{|\tg|} h^\varepsilon( i \psi, \eta_{ab} \gamma^a \beta^b_\mu D^\mu_\h\psi)\bigr) \label{eom csm cst lorentz '}
	\end{align}
	One can directly observe that taking $S'_M$ instead of $S_M$ changed the gravitational coupling such that both $G'$ and $\Lambda_G$ are now proportional to $\chi'$ with an added contribution of the matter fields to $\chi'$ when compared to $\chi$.
	
	In the following table we compare the behavior of $\Lambda_G$, $\Lambda_{\text{vac}}$ and $G$ depending on the choice of matter action:
	\begin{align} \label{summary quasi Lorentz} 
		\left\{
		\begin{array}{ll}
			S_T= S_G + S_M  &\Leftrightarrow \Lambda_G = \dfrac{\chi}{4} \text{ and }
			\left\{
			\begin{array}{ll}
					\Lambda_{\text{vac}} \\
					G
			\end{array}
			\right\} \propto \dfrac{1}{\chi} \\[2mm]
			S'_T= S_G + S'_M   &\Leftrightarrow \Lambda'_G = \dfrac{\chi'}{4}  \text{ and }
			\left\{
			\begin{array}{ll}
				\Lambda'_{vac} = 8 \pi G' \rho_{vac}  \\
				G'=G_0 \tilde{\chi'}
			\end{array}
			\right\} \propto \chi'
		\end{array}
		\right.
	\end{align}

	Thus, if future measurements confirm the
	Hubble tension, a possible explanation could reside in considering a moving geometry with spacetime-dependent scalar fields like $k$,$k'$ and $\ell$ instead of constants. Not including the mutation degrees of freedom ($k,k',\ell$) in the matter action through minimal coupling makes the gravitational coupling spacetime dependent too with an inverse dependence on $\chi$ compared to $\Lambda_G \propto \chi$. This model can be compared to the $\alpha = -2,~ \beta = -4$ model of \cite{sengupta2025cosmological} with added topological terms and Holst action with parameter $y$ that provides a coupling between non radiative matter and pseudo-scalar curvature. On the other hand taking into account the mutation in the matter sector via the $S'_M$ matter action gives us $\Lambda'_G \propto G' \propto \chi'$ with $\chi'$ now sporting an additional dependence on the matter fields. This last dependence $\Lambda'_G \propto G' \propto \chi'$ with a $G$ decaying over cosmological time is not favored by observations as argumented in \cite[p.9]{sengupta2025cosmological} where it corresponds to a model with $\alpha >0$.  We think the formalism presented here, with notably the matter action $S^f_M$ provides a good way to study which dependence on $k,k',\ell$ best fits current data and shall be explored in the future.
	
	The only kinetic term for the gauge contribution $\Lambda_G$ to dark energy in action \eqref{action Lorentz} originates from the Pontryagin number associated to the total curvature $\bOm$ which is inter alia linked to the Nieh-Yan term in the action. An important thing to note is that the kinetic term is null if $T=0$, as is for example the case for General Relativity (GR).

    Examining the EOM's \eqref{eom frame Lorentz quasi}-\eqref{eom spin Lorentz quasi} we derive that choosing a total action $S = S_G + S^f_M$ with $f(k,k',\ell) = \dfrac{\tk \tk'}{\tl}$ allows to retrieve a spacetime dependent dark energy $\Lambda_G$ with the usual Newton constant $G_0$ and vacuum contribution $\Lambda_{vac} = 8 \pi G_0 \rho_{vac}$.
    
    A more detailed study of the solutions of the equations of motion and notably their stability under perturbations shall be performed in subsequent work, notably for the specific case involving a crossing between positive/negative $\Lambda_G$.
    
    In the next section we shall now study a generalization of the previous geometry, based on the Möbius group, which offers a natural geometry to describe conformal gravity with the Schouten, Cotton and Weyl tensors.
		
		\section{Moving Lorentz$\times$Weyl Cartan geometry}
		
		\label{section mobius}
		
			Another interesting example of geometry is given by the Möbius group $G=SO(4,2)/\{\pm I\}$ defined by the relation $g^T N g = N$, $g \in G$, with 
	\begin{align*}
		N = \begin{pmatrix}
			0 & 0 & -1 \\
			0 & \eta & 0 \\
			-1 & 0 & 0
		\end{pmatrix} 
	\end{align*}  
    
	The Lie algebra of $G$ is $\g = \mathfrak{so}(4,2) = \g_{-1} \oplus \g_0 \oplus \g_{1} $ with $\g_{-1} \simeq \mathbb{R}^{3,1} $, $\g_0 \simeq \mathfrak{co}(3,1) = \mathfrak{so}(3,1) \oplus \mathbb{R}$ and $ \g_1 \simeq \mathbb{R}^{3,1*}$. The Lie algebra $\g$ is a graded Lie algebra \cite{kobayashi_transformation_1995} with commutation relations:
	\begin{align}
		& [ \g_0 , \g_0 ] \subset \g_0, \qquad [ \g_0 , \g_{-1} ] \subset \g_{-1},  \qquad [ \g_1 , \g_0 ] \subset \g_1, \qquad [\g_{-1},\g_1] \subset \g_0 \notag \\[-4mm]
		& \label{eq:gradedLa} \\
		& \qquad \qquad \qquad \qquad \qquad \; \; [ \g_{-1} , \g_{-1} ] = [ \g_{1} , \g_{1} ] = 0. \notag
	\end{align} 

Let us consider $H=CO(3,1)$ as structure group for the Cartan geometry to be discussed in this part.

	The following decomposition of the graded Lie algebra
    \begin{align} \label{eq:graded-equiv}
    \g= \overbrace{\g_0}^{\h}  \oplus \overbrace{\g_{1} \oplus \g_{-1}}^{\m} \simeq \overbrace{\mathfrak{co}(3,1)}^{\mathfrak{so}(3,1) \oplus \mathbb{R}} \oplus {\mathbb{R}^{3,1}}^* \oplus \mathbb{R}^{3,1}
    \end{align}
    of the Möbius group is reductive with respect to $H$ and is also a symmetric Lie algebra $\g = \mathfrak{co}(3,1) \oplus \m$  where $\m= \g_{1} \oplus \g_{-1} \simeq {\mathbb{R}^{3,1}}^* \oplus \mathbb{R}^{3,1}$. The commutation relations (see \eqref{eq:gradedLa} and \cite[Example 5.2 and Chapter XI]{kobayashi_foundations_1996}) are given by
	\begin{align}
		[ \g_0 , \g_0 ] \subset \g_0, \qquad [ \g_0 , \m ] \subset \m,  \qquad [ \m,\m] \subset \g_0\ 
	\end{align} 
    making $\g=\g_0\oplus \m$ a symmetric Lie algebra.
    
According to the graded Lie algebra decomposition of $\g$, let  
    \[
    \varphi_\g = \varphi_1\oplus \varphi_0\oplus \varphi_{-1} \in \g= \g_{1}\, \oplus\, \g_0\, \oplus\, \g_{-1}
    \] 
for which we adopt the matrix presentation\footnote{Accordingly, the Killing form on $\g = \h\oplus\m$ is computed to be 
   \[
   K_\g(\varphi,\varphi') = \Tr(\ad_{\varphi}\circ \ad_{\varphi'}) = n \Tr(cc')+ 2n\big( zz'+\eta(b,a')+\eta(a,b')\big)
   \]
   where $n$ corresponds to the one coming from $\dim \mathfrak{so}(n-1,1) = n(n-1)/2$.}
	\begin{align}\label{eq:graded-decomp}
		\varphi_\g = \varphi_1\oplus \varphi_0\oplus \varphi_{-1} = \begin{pmatrix}
			z & \bar{a}	& 0 \\
			b & c & a \\
			0 & \bar{b} & -z
		\end{pmatrix}
		=
		\begin{pmatrix}
			0 & \bar{a} & 0 \\
			0 & 0 & a \\
			0 & 0 & 0
		\end{pmatrix}
		+
		\begin{pmatrix}
			z & 0 & 0 \\
			0 & c & 0 \\
			0 & 0 & -z
		\end{pmatrix}
		+
		\begin{pmatrix}
			0 & 0 & 0 \\
			b & 0 & 0 \\
			0 & \bar{b} & 0
		\end{pmatrix}
	\end{align}
with $a,b \in \mathbb{R}^{3,1}$, $c \in \mathfrak{so}(3,1)$, $z \in \mathbb{R}$ and $\bar{a}= a^T \eta, \bar{b}= b^T \eta \in \mathbb{R}^{*3,1}$. It is worthwhile to notice that $E = \begin{pmatrix}
    1 & 0 & 0\\ 0 & 0 & 0 \\ 0 & 0 & -1
\end{pmatrix}\in \g_0$ the generator for dilation defines the grading of $\g$ by $\g_j = \{\varphi\in\g \text{ s.t. } [E,\varphi] = j \varphi \}, j=0,\pm 1$.
Moreover, through the isomorphism $\g\simeq {\mathbb{R}^{3,1}}^* \oplus \mathfrak{co}(3,1) \oplus \mathbb{R}^{3,1}$, one can also write \cite{kobayashi_transformation_1995}
\[
\varphi_\g = \varphi_1\oplus \varphi_{\mathfrak{co}} \oplus \varphi_{-1} = \varphi_1\oplus (\varphi_{\mathfrak{so}} - \varphi_{\mathbb{R}}\mathds{1}_4) \oplus \varphi_{-1} =\bar{a} \oplus (c - z \mathds{1}_4) \oplus b.
\]
		
		\subsection{Moving Lorentz$\times$Weyl spacetimes as mutated $G'/H$ geometries}
		
    	We consider the mutation $\g'=\mu(\g)=\mu(\h\oplus\m) = \h\oplus\mu(\m)=\h\oplus\m'$ given in \cite{thibaut2024gravity} by:
\begin{align}
		\label{mutation Möbius}
		\varphi_\g =
		\begin{pmatrix}
            z & \bar{a} & 0 \\
			b  & c	& a \\
			0 & \bar{b}  & -z 
		\end{pmatrix}
		\mapsto
		\varphi_{\g'} = \mu(\varphi_\g)=
		\begin{pmatrix}
			z & \bar{\gamma}_2 \bar{a} & 0 \\
			\gamma_1 b   & c & \gamma_2 a \\
			0 & \bar{\gamma}_1 \bar{b}  & -z 
		\end{pmatrix} \in \g'
\end{align}
where the scalar mutation parameters $\gamma_1,\bar{\gamma}_1\gamma_2,\bar{\gamma}_2$ are arranged such that $\mu$ is a linear isomorphism, in particular, $\gamma_1 \bar{\gamma}_2= \bar{\gamma}_1\gamma_2$. Thus, it is merely a change of scale within the symmetric Lie algebra which preserves the splitting $\m=\g_{-1}\oplus\g_1$ according to the graded structure.

\smallskip
Let us choose $\gamma_1=\bar{\gamma}_1=\dfrac{1}{\ell}$, $\gamma_2=\bar{\gamma}_2 = \dfrac{1}{\ell'}$ (where $\ell$ and $\ell'$ are non-null $0$-form scalar fields).

    The Lie bracket on $\g'$ is thus given by the commutator:
	\begin{align*}
		[\varphi,\varphi'] & =
		%%%%%%%%%%%%%%%
		\begin{bmatrix}
			\begin{pmatrix}
				z & \dfrac{1}{\ell'} \bar{a}	& 0 \\
				\dfrac{1}{\ell} b & c & \dfrac{1}{\ell'} a \\
				0 & \dfrac{1}{\ell} \bar{b} & -z
			\end{pmatrix}
			, %%%%%%%%%%%%%%%
			\begin{pmatrix}
				z' & \dfrac{1}{\ell'} \bar{a}'	& 0 \\
				\dfrac{1}{\ell} b' & c' & \dfrac{1}{\ell'} a' \\
				0 & \dfrac{1}{\ell} \bar{b}' & -z'
			\end{pmatrix}
		\end{bmatrix} \numberthis \\
		%%%%%%%%%%%%%%%%%%%%%
		& =
		\begin{pmatrix}
			\dfrac{1}{\ell\ell'} (\bar{a} b' - \bar{a}' b)  & \dfrac{1}{\ell'} (z\bar{a}' - z' \bar{a}	+ \bar{a} c' - \bar{a}' c) & \dfrac{1}{\ell'^2} (\bar{a} a' - \bar{a}' a) \\
			%%%%%%%%%%%%%%%%%%%%
			\dfrac{1}{\ell} (z' b - z b' + c b' - c' b) & [c,c'] + \dfrac{1}{\ell\ell'} (b \bar{a}' - b' \bar{a} + a \bar{b}' - a' \bar{b}) & \dfrac{1}{\ell'} (c a' - c' a - z' a + z a') \\
			%%%%%%%%%%%%%%%%%%%%
			\dfrac{1}{\ell^2} (\bar{b} b' - \bar{b}' b) & \dfrac{1}{\ell} (\bar{b} c' - \bar{b}' c - z \bar{b}' + z' \bar{b}) & \dfrac{1}{\ell\ell'} (\bar{b} a' - \bar{b}' a)
		\end{pmatrix} \\
		& = 
		\begin{pmatrix}
			\dfrac{1}{\ell\ell'} (\bar{a} b' - \bar{a}' b)  & \dfrac{1}{\ell'} (z\bar{a}' - z' \bar{a}	+ \bar{a} c' - \bar{a}' c) & 0 \\
			%%%%%%%%%%%%%%%%%%%%
			\dfrac{1}{\ell} (z' b - z b' + c b' - c' b) & [c,c'] + \dfrac{1}{\ell\ell'} (b \bar{a}' - b' \bar{a} + a \bar{b}' - a' \bar{b}) & \dfrac{1}{\ell'} (c a' - c' a - z' a + z a') \\
			%%%%%%%%%%%%%%%%%%%%
			0 & \dfrac{1}{\ell} (\bar{b} c' - \bar{b}' c - z \bar{b}' + z' \bar{b}) & \dfrac{1}{\ell\ell'} (\bar{b} a' - \bar{b}' a)
		\end{pmatrix}
	\end{align*}

		\subsection{Cartan connection and curvature}

	A Cartan connection $\varpi$ with this $G/H$ moving Cartan geometry can be decomposed as :
	
	\begin{align}
		\varpi = \varpih \oplus \varpip = \varpi_{1} \oplus \varpi_0 \oplus \varpi_{-1}
		=& \varpi_{1}^a
		\begin{pmatrix}
			0 & \bar{e}_a & 0 \\
			0 & 0 & e_a \\
			0 & 0 & 0
		\end{pmatrix}
		+
		\begin{pmatrix}
			\lambda & 0 & 0 \\
			0 & A & 0 \\
			0 & 0 & - \lambda
		\end{pmatrix}
		+ \varpi_{-1}^a
		\begin{pmatrix}
			0 & 0 & 0 \\
			e_a & 0 & 0 \\
			0 & \bar{e}_a & 0
		\end{pmatrix} \\
		=&
		\begin{pmatrix}
			\lambda & \dfrac{1}{\ell'} \bar{\alpha} & 0 \\
			\dfrac{1}{\ell} \beta & A & \dfrac{1}{\ell'} \alpha \\
			0 & \dfrac{1}{\ell} \bar{\beta} & - \lambda
		\end{pmatrix}
	\end{align}
	
	where $A= \dfrac{1}{2} A^{ab} J_{ab} =\dfrac{1}{2} A^{ab}_\mu dx^\mu \otimes J_{ab} $ corresponds to the spin connection and 	$\beta = \beta^a e_a = \beta^a_\mu dx^\mu \otimes e_a$ corresponds to the tetrad or soldering form used to define the metric $g= \eta(\beta,\beta)$.
	$\alpha = \alpha^a e_a = \alpha^a_\mu dx^\mu \otimes e_a$ can be interpreted as a secondary tetrad.
	Finally $\lambda = \lambda_\mu dx^\mu$ is the $\mathbb{R}$-valued (dilations) part of the connection.
	Also : $\bar{\alpha} = \alpha^a \bar{e}_a = \alpha^a_\mu dx^\mu \otimes \bar{e}_a$ and $\bar{\beta} = \beta^a \bar{e}_a = \beta^a_\mu dx^\mu \otimes \bar{e}_a $.

	The expression of the curvature $\bOm = d\varpi + \dfrac{1}{2} [\varpi,\varpi] $ of $\varpi$ is :
	
	\begin{align}
		\bOm & = \bOmh \oplus \bOmp =
		\begin{pmatrix}
			f & \dfrac{1}{\ell'}\bar{\Pi} & 0 \\
			\dfrac{1}{\ell}\Theta & F & \dfrac{1}{\ell'}\Pi \\
			0 & \dfrac{1}{\ell}\bar{\Theta} & - f
		\end{pmatrix}
		=
		\begin{pmatrix}
			f & \dfrac{1}{\ell'}\bar{\Pi} & 0 \\
			0 & F & \dfrac{1}{\ell'}\Pi \\
			0 & 0 & - f
		\end{pmatrix} 
		+
		\begin{pmatrix}
			0 & 0 & 0 \\
			\dfrac{1}{\ell}\Theta & 0 & 0 \\
			0 & \dfrac{1}{\ell}\bar{\Theta} & 0
		\end{pmatrix} 
		\\
		& = \resizebox{\linewidth}{!}{%
			$	\displaystyle
			\begin{pmatrix}
				d\lambda +\dfrac{1}{\ell \ell'}  \bar{\alpha} \wedge \beta  & \dfrac{1}{\ell'} (d\bar{\alpha} - \dfrac{d\ell'}{\ell'}\wedge \bar{\alpha} +\lambda \wedge \bar{\alpha} + \bar{\alpha} \wedge A) & 0 \\
				\dfrac{1}{\ell} (d\beta - \dfrac{d\ell}{\ell}\wedge \beta - \lambda \wedge \beta + A \wedge \beta) & dA + \dfrac{1}{2} [A,A] + \dfrac{1}{\ell \ell'} (\beta \wedge \bar{\alpha} + \alpha \wedge \bar{\beta}) & \dfrac{1}{\ell'} (d \alpha - \dfrac{d\ell'}{\ell'}\wedge \alpha + A \wedge \alpha + \lambda \wedge \alpha) \\
				0 & \dfrac{1}{\ell} (d\bar{\beta} - \dfrac{d\ell}{\ell}\wedge \bar{\beta} + \bar{\beta} \wedge A - \lambda \wedge \bar{\beta}) & - d\lambda + \dfrac{1}{\ell \ell'} \bar{\beta} \wedge \alpha
			\end{pmatrix} $}
	\end{align}
	with: 
	\begin{align*}
		F= R + \dfrac{1}{\ell \ell'} \phi \qquad 
		\phi= \beta \wedge \bar{\alpha} + \alpha \wedge \bar{\beta} \qquad
		f = d\lambda + \dfrac{1}{\ell \ell'} \bar{\alpha} \wedge \beta \qquad	
	\end{align*}  %\\[2mm]
		\vspace{-0.4cm}
	\begin{align*}
		\Pi = d\alpha - \dfrac{d\ell'}{\ell'}\wedge \alpha + A \wedge \alpha + \lambda \wedge \alpha
	\end{align*} 
	\vspace{-0.4cm}
	\begin{align*}
		\Theta = d\beta - \dfrac{d\ell}{\ell}\wedge \beta - \lambda \wedge \beta + A \wedge \beta =  T - \dfrac{d\ell}{\ell}\wedge \beta - \lambda \wedge \beta
	\end{align*} 
	As shown in \cite[p.55]{francois_reduction_2014-1}, when $\ell$ and $\ell'$ are constants it is then possible to make contact with conformal gravity by imposing $\lambda = 0$ and a normal Cartan geometry (defined by $Ricci(\bOm)= 0$ and $\bOmp = 0$) to retrieve:
	\begin{itemize}
		\item $A$ as the torsion free ($\bOmp = \Theta = T = 0$) Levi-Civita connection
		\item $\alpha_{b,c} = \dfrac{1}{2} (R^{a}_{~b,ac} - \dfrac{1}{6} \eta_{bc} R)$ as the Schouten tensor
		\item and $\Pi$ and $F$ as respectively the Cotton and Weyl tensors.
	\end{itemize}

	\begin{remark}
	    In the general case one needs the conditions:
	    \[
	    \bar{\gamma_1} d\gamma_1 = \gamma_1 d\bar{\gamma_1} \text{ and } \bar{\gamma_2} d\gamma_2 = \gamma_2 d\bar{\gamma_2}
	    \]
	    for the curvature $\bOm$ to stay inside the Lie algebra. These conditions are however unnecessary for the Bianchi identities to hold.
	\end{remark}
	
	The Bianchi identities are given by $D\bOm = d\bOm + \varpi \wedge \bOm - \bOm \wedge \varpi = 0$ which splits into the identities :
	
	\begin{align}
		\label{Bianchi Möbius}
		& (D\bOm)_{-1} = d(\dfrac{\Theta}{\ell}) + \dfrac{1}{\ell} \bigl( (A - \lambda) \wedge \Theta + (f -F) \wedge \beta \bigr) = 0 \\
		& (D\bOm)_{\mathfrak{so}(3,1)} = \underbrace{dR + [A,R]}_{=0} + d(\dfrac{\phi}{\ell\ell'}) + \dfrac{1}{\ell \ell'} ( [A,\phi] + \beta \wedge \bar{\Pi} - \Pi \wedge \bar{\beta} + \alpha \wedge \bar{\Theta} - \Theta \wedge \bar{\alpha} ) = 0 \\
		& (D\bOm)_{\mathbb{R}} = - \mathds{1} \bigl(df + \dfrac{1}{\ell \ell'} (\bar{\Theta} \alpha - \bar{\beta} \Pi) \bigr) = 0 \\
		& (D\bOm)_{1} = d(\dfrac{\Pi}{\ell'})  + \dfrac{1}{\ell'} \bigl( (A + \lambda) \wedge \Pi -(F + f) \wedge \alpha \bigr) = 0
	\end{align}
		They stay true even if the mutation condition $\gamma_1 \bar{\gamma}_2 = \bar{\gamma}_1 \gamma_2$ isn't fulfilled
	or if we consider the theory on the submanifold $\cm \subset \tilde{\cm}$ defined by the identification $\alpha=\dfrac{k}{2} \beta$ as was done in \cite{thibaut2024gravity} except for the fact that we now define $k$ as a 0-form $\mathbb{R}$-valued scalar field. 
	Additionally, one has:
		\begin{prop}
				\label{submanifold} The subset of vector fields
				$\mathcal{D}_k = \{ X \in \GamTtM | \alpha (X) = \dfrac{k}{2} \beta (X) 
                \}$ defines an integrable submanifold of dimension four $\mathcal{M} \subset\widetilde{\mathcal{M}}$.
		\end{prop}
		
		\begin{proof} Indeed, 
				$\mathcal{D}_k = \{ X \in \GamTtM | \alpha (X) = \dfrac{k}{2} \beta (X) 
                \}$  provides us with the vielbein $\alpha - \dfrac{k}{2} \beta$ which corresponds to the four Pfaff forms defining the submanifold $\cm$ such that $\alpha - \dfrac{k}{2} \beta = 0$ on $\cm$. 
                For all $X,Y \in \mathcal{D}_k$ one has $ (\Pi - \lambda \wedge \alpha + \dfrac{d\ell'}{\ell'}\wedge \alpha) (X,Y) = \dfrac{k}{2} \bigl(T +\dfrac{dk}{k}\wedge \beta \bigr)(X,Y) $ which yields
				\begin{align*}
					d(\alpha - \dfrac{k}{2} \beta) = - A \wedge (\alpha - \dfrac{k}{2} \beta).
				\end{align*}
				Hence, by the Frobenius theorem \cite[Prop.5.3, p.81]{sharpe_differential_1997} the corresponding distribution defines $\cm$ as an integrable submanifold of $\widetilde{\mathcal{M}}$. Therefore, there is a subclass of conformal Cartan connections characterizing $\cm$; this subclass can be reached through this kind of "gauge fixing".
		\end{proof}

In the next sections we shall now see that considering action \eqref{eq:action_start} in the general case $\alpha \neq \dfrac{k}{2} \beta$ leads to new terms involving both dilations and variations of the fields $\ell,\ell'$. In particular we show the action allows to retrieve the previous Lorentzian action up to an additional dilational topological term when considering the $4$-dimensional submanifold $\cm$ defined by the specific choice of tetrads $\alpha = \dfrac{k}{2} \beta$.
 We also describe a matter action that incorporates the dilational degrees of freedom and compute the associated equations of motion. 

	\subsection[Moving Lorentz$\times$Weyl deformed action]{Moving Möbius/conformal deformed action}
	\label{Moving Möbius/conformal deformed action}

	The trace $Tr : \g \times \g \rightarrow \mathbb{R}$ will be used to compute the Pontryagin number associated to $\bOm$. Here $Tr(\gamma,\gamma') = - \eta_{ac} \eta_{bd}\gamma_{\mathfrak{so}}^{ab} \gamma_{\mathfrak{so}}^{'cd} + 2 \gamma_{\mathbb{R}} \gamma'_{\mathbb{R}} + \dfrac{2}{\ell\ell'} \eta_{ab} (\gamma_{-1}^a \gamma_{1}^{'b} + \gamma_{1}^a \gamma_{-1}^{'b})$. 
	
	Action \eqref{eq:action_start} for this moving geometry reads:
	
	\begin{align*}
		\label{Mobius action dyn}
		S_G[\varpi] 
		& =  \int_{\tilde{\cm}}  \Bigl( e \Pf (\dfrac{\bOmh}{2\pi})
		- \dfrac{r }{8\pi^2} \Tr (\bOm  \bOm)
		- \dfrac{y }{8\pi^2} \Tr (F  F) \Bigr) \\
		%%%%%%%%%%%%%
		& = \int_{\tilde{\cm}} \Bigl(
		\dfrac{e }{2 (4\pi)^2} \varepsilon_{abcd} F^{ab} \wedge F^{cd}  
		- \dfrac{r}{8\pi^2} \Tr ( \bOm \wedge \bOm )
		- \dfrac{y}{8\pi^2} \Tr ( F \wedge F ) 
		\Bigr) \\
		%%%%%%%%%%%%%%%%
		&= \int_{\tilde{\cm}} \Biggl(
		\dfrac{e }{2(4\pi)^2} \varepsilon_{abcd} \bigl(  R^{ab} \wedge R^{cd} + \dfrac{2}{\ell \ell'} R^{ab} \wedge \phi^{cd} + \dfrac{1}{(\ell \ell')^2} \phi^{ab} \wedge \phi^{cd} 
		\bigr) \numberthis \\
		& \qquad \quad + \dfrac{1}{8 \pi^2} \Bigl( (r+y) \bigl(R^{ab} \wedge R_{ab} + \dfrac{2}{\ell \ell'} R^{ab} \wedge \phi_{ab} + \dfrac{1}{(\ell \ell')^2} \underbrace{\phi^{ab} \wedge \phi_{ab}}_{=0}\bigr) - \dfrac{4r}{\ell \ell'} \eta_{ab} \Pi^a \wedge \Theta^b - 2rf \wedge f
		\Bigr)
		\Biggr) 
	\end{align*}
	where $\tilde{\cm}$ is the manifold defined by the moving Cartan geometry corresponding to the quotient $G/H$ of dimension $8$.
	Upon identifying $\alpha = \dfrac{k}{2} \beta$ (as done in \cite{thibaut2024gravity}), $\tilde{\cm}$ reduces to the integrable submanifold $\cm\subset \tilde{\cm}$ of dimension $4$ and the curvature entries transform to
    \[
    \Theta = T - ( \dfrac{d\ell}{\ell} + \lambda ) \wedge \beta, \quad
    \Pi
    = \dfrac{k}{2} \bigl(T + (\lambda + \dfrac{dk}{k} - \dfrac{d\ell'}{\ell'} \bigr) \wedge \beta),
    \quad \phi = k\beta\wedge\bar{\beta} = k\ell^2 \xi,
    \quad F= R + \dfrac{k\ell}{\ell'} \xi \quad \text{and}
    \quad f=d\lambda
    \]
	The action on the submanifold $\cm$ is:
	\begin{align*}
		S_G[\varpi] 
		%%%%%%%%%%%%%%%%%%%%%%%%%%%%%%%%%%%%%%%%%%%%%%%%%%%%%%%%%%%%%%%%%%%%%%%%% 		%%%%%%%%%%%%%%%%%%%%%%%%%%%%%%%%%%%%%%%%%%%%%%%%%%%%%%%%%%%%%%%%%%%%%%%%%
		= \int_{\mathcal{M}} &\Bigl(
		\overbrace{ \dfrac{k }{4\pi^2\ell\ell'} \big( \dfrac{e}{4} \underbrace{R^{ab} \wedge \beta^c\wedge\beta^d \varepsilon_{abcd}}_{Palatini} + y R^{ab} \wedge \beta_a\wedge\beta_b \big)}^{Holst}
		+ \overbrace{\dfrac{e k^2}{32\pi^2(\ell\ell')^2}  \beta^a\wedge\beta^b \wedge \beta^c\wedge\beta^d \varepsilon_{abcd}}^{Bare~Cosmological~axpansion} \nonumber
		\\
		%%%%%%%%
		& \qquad + \overbrace{\dfrac{r +y}{8 \pi^2} R^{ab} \wedge R_{ab}}^{Pontrjagin}
		+ \overbrace{\dfrac{e }{2(4\pi)^2} R^{ab} \wedge R^{cd} \varepsilon_{abcd}}^{Euler} 
		\overbrace{- \dfrac{r k }{4 \pi^2 \ell\ell'} (T^a \wedge T_a - R^{ab} \wedge \beta_a\wedge\beta_b)}^{Nieh-Yan}
		\\[2mm]
		&\overbrace{- \dfrac{r}{4\pi^2} d\lambda \wedge d\lambda}^{Kinetic ~ term ~ for ~ \lambda}
		\overbrace{- \dfrac{kr}{4\pi^2\ell\ell'} (\dfrac{dk}{k} - \dfrac{d\ell}{\ell} - \dfrac{d\ell'}{\ell'}) \wedge \beta^a \wedge T_a}^{Kinetic~term~for~k,~\ell~and~\ell'} \numberthis \Bigr) 
	\end{align*}
	
	\subsection[Möbius/conformal matter actions]{Möbius/conformal matter actions}
	\label{Möbius/conformal matter actions}

	Choosing the symmetric decomposition $ \g= \g_0 \oplus \m $ with $\m = \g_1 \oplus \g_{-1} = \mathbb{R}^{3,1*} \oplus \mathbb{R}^{3,1}$ and $\g_0 = \mathfrak{co}(3,1)$ we compute the Killing metric: 
	
	$$K_\g (\varphi,\varphi')
	= \overbrace{- 4 \eta_{ac} \eta_{bd}\varphi_{\mathfrak{so}}^{ab} \varphi_{\mathfrak{so}}^{'cd} + 8 \varphi_{\mathbb{R}} \varphi'_{\mathbb{R}}}^{K_0} + \overbrace{ \dfrac{8}{\ell\ell'} \eta_{ab} ( \varphi_{-1}^a \varphi_{1}^{'b} + \varphi_{1}^a \varphi_{-1}^{'b} )}^{K_\m}  $$
	giving us the metrics (after $\alpha = \dfrac{k}{2} \beta$ identification): 
	\begin{align}
		& h(\varphi,\varphi') = \kappa K_\g (\varphi,\varphi') \\
		&g (X,Y)  = \zeta (\varpim^*h) (X,Y) = \kappa \zeta (\varpim^*K_\m) (X,Y) = \dfrac{8k \zeta \kappa }{\ell\ell'} \eta \bigl( \beta(x) , \beta(Y) \bigr) = \dfrac{8k \zeta \kappa }{\ell\ell'} \tilde{g} (X,Y) .
	\end{align}
	where $\kappa(x)$ and $\zeta(x)$ are functions depending on the point $x \in \cm$. 
	
	Let $\omega \in \Omega^r(\mathcal{M},\g) $ be an $r$-form on $\mathcal{M}$ with values in $\g$.
	Its local trivialisation in a given chart is :
	\begin{align}
		\omega= \dfrac{1}{r!} \omega_{\mu_1\mu_2...\mu_r}dx^{\mu_1}\wedge dx^{\mu_2}\wedge ... \wedge dx^{\mu_r}
	\end{align}
	
	From the metric $g$ we define the Hodge star operator $*$ that acts on $\omega$ as :
	\begin{align}
		*\omega = & \dfrac{1}{r!} \sqrt{|det(g)|}\,\omega_{\mu_1...\mu_r} 
		g^{\mu_1 \nu_1}...g^{\mu_r \nu_r} \varepsilon_{\nu_1...\nu_m} dx^{\nu_{r+1}}\wedge ... dx^{\nu_m}
	\end{align}
	
	The most conservative choice of matter action is then:	
	\begin{align}
		\label{Matter action 4}
		S_M[A,\beta] & = \int_{\mathcal{M}} h^\varepsilon\bigl( i \psi, h_{\m,ab} \gamma^a \varpim^b \wedge *D\psi\bigl)
		=\int_{\mathcal{M}} \dfrac{64k^2\kappa^2\zeta}{(\ell\ell')^2} h^\varepsilon\bigl( i \psi, \eta_{ab} \gamma^a \beta^b \wedge \tilde{*}D\psi\bigl)
		= \mathrm{Re} S_D
	\end{align}
	with $D = (\partial_\mu + \dfrac{1}{4} A_{ab,\mu} \gamma^a \gamma^b) dx^\mu$ the covariant derivative relative to the Lorentz valued part of the connection already used in \eqref{Matter action 3} and where we set $\kappa^2\zeta = \tfrac{(\ell \ell')^2}{384k^2}$ such that the overall factor is arranged to stick with the Dirac action given in \cite[p.197]{gockeler_differential_1987},
	with $\mathrm{Re} S_D$ corresponding to the real part of the Dirac action.
	
	As has been described in \cite{thibaut2024gravity} we can also consider the action:	
	\begin{align}
		\tS_M[A,\beta,\lambda] & =
		\int_{\mathcal{M}} h^\varepsilon\bigl( i \psi, h_{ab} \gamma^a \varpim^b \wedge *D_0\psi\bigl) \\
		%%%%%%
		& = \overbrace{S_M}^{\mathrm{Re} S_D} - 
		\dfrac{1}{2}\int_{\mathcal{M}} h^\varepsilon\bigl( i \psi, h_{ab} \gamma^a \varpim^b \wedge *(\lambda\psi)\bigl).
	\end{align}
	where $D_0 = d + \Phi_\varepsilon (\varpiz) = d + \Phi_\varepsilon (A - \lambda \mathds{1}_4)
	= (\partial_\mu + \dfrac{1}{4} A_{ab,\mu} \gamma^a \gamma^b - \dfrac{1}{2} \lambda_\mu) dx^\mu$ is the covariant derivative taking into account the full action of the $ \g_0 = \mathfrak{co}(3,2) $-valued part of the connection $\varpiz$.
	
	However what we deem to be the most natural matter actions modeled after the Dirac action are:
	\begin{align}
		S'_M[A,\beta,k,\ell,\ell'] = \int_{\mathcal{M}} h^\varepsilon\bigl( i \psi, h_{ab} \gamma^a \varpim^b \wedge *D\psi\bigl) 
		=\int_{\mathcal{M}} \dfrac{64k^2\kappa^2\zeta}{(\ell\ell')^2} h^\varepsilon\bigl( i \psi, \eta_{ab} \gamma^a \beta^b \wedge \tilde{*}D\psi\bigl)
		= \dfrac{\tk^2}{(\tl\tl')^2} S_M
	\end{align}
	and
	\begin{align}
		\tS'_M[A,\beta,\lambda,k,\ell,\ell'] & = \int_{\mathcal{M}} h^\varepsilon\bigl( i \psi, h_{ab} \gamma^a \varpim^b \wedge *D_0\psi\bigl) \nonumber \\
		%%%%%%
		& = S'_M - \dfrac{1}{2} \int_{\mathcal{M}} h^\varepsilon\bigl( i \psi, h_{ab} \gamma^a \varpim^b \wedge *(\lambda\psi)\bigl) =  \dfrac{\tk^2}{(\tl\tl')^2} \tS_M
	\end{align}
	where this time we set $\kappa^2\zeta = \tfrac{(\ell_0 \ell_0')^2}{384k_0^2}$ with $k_0=k(x_0)$ and $\tilde{k} = \dfrac{k}{k_0}$, $x_0 \in \mathcal{M}$ (the same goes for $\ell$ and $\ell'$). This way $S'_M$ and $\tS'_M$ respectively coincide with $S_M$ and $\tS_M$ at a reference point $x_0$ but not in general, thus taking into account the degrees of freedom hidden in $k$, $\ell$ and $\ell'$.
    
    As in section \ref{Moving Lorentzian geometry actions} one may also opt for other non-trivial matter actions expressed as:
    \begin{align}
        \label{f matter action Mobius}
        S^f_M[A,\varpim] &= f(k,k',\ell) \mathrm{Re}S_D \; \text{ or } \; \tilde{S}^f_M[A, \lambda,\varpim] = f(k,k',\ell) \tilde{S}_M
    \end{align}
    where $f(k,k',\ell)$ is a function of the mutation scalar fields.
	
	Let $ \int_\mathcal{M} \dfrac{v k}{2\pi^2\ell \ell'} \eta_{ab} T^a \wedge \beta^b \wedge \lambda $ be an interaction between torsion and dilations with $v$ a constant.
	We denote by $S'_G = S_G + \int_\mathcal{M} \dfrac{v k}{2\pi^2\ell \ell'} \eta_{ab} T^a \wedge \beta^b \wedge \lambda $ the sum of the deformed topological gauge action $S_G$ with this interaction.
	
	We now compute the equations of motions associated to the 4 total actions $S_T = S'_G + S_M$, $\tS_T = S'_G + \tS_M$, $S'_T = S'_G + S'_M$ and $\tS'_T = S'_G + \tS'_M$.
	
	The first 3 EOM's on $\beta$, $A$ and $\lambda$ are:
	\begin{align}
		\dfrac{\delta \mathcal{L}'_G }{\delta \beta^c} &= 
		%%%%%
		\dfrac{ek}{8\pi^2\ell\ell'} ( R^{ab} + \dfrac{ k }{\ell\ell'} \beta^a \wedge \beta^b ) \wedge \beta^d \varepsilon_{abcd} \nonumber \\ 
		&+ \dfrac{k}{2\pi^2\ell\ell'} \biggl( y R^{ab} \wedge \beta_b \delta_{ac} - v \Bigl( \bigl(( \dfrac{dk}{k} - \dfrac{d\ell}{\ell} - \dfrac{d\ell'}{\ell'} ) \wedge \lambda + d\lambda \bigr) \wedge \beta_c - 2 \lambda \wedge T_c \Bigr) \biggr) \\
		%%%%%%%%
		&=  \left\{
		\begin{array}{ll}
			\tau_c \text{ for } S_M \\
			\dfrac{\tk^2}{(\tl\tl')^2} \tau_c \text{ for } S'_M \\
			\tau_c + \dfrac{1}{12} h^\varepsilon\bigl( i \psi, \gamma^a \eta_{ac} \tilde{*}(\lambda \psi) \bigl) \text{ for } \tS_M \\
			\dfrac{\tk^2}{(\tl\tl')^2} \Bigl(\tau_c + \dfrac{1}{12} h^\varepsilon\bigl( i \psi, \gamma^a \eta_{ac} \tilde{*}(\lambda \psi) \bigl)\Bigr) \text{ for } \tS'_M 
		\end{array}	
		\right. \nonumber \\[2mm]
		%%%%%%%%%%%%%%%%%%%%%%%%%%%%%%%%%%%%%%%%%%%%%%%%%%%%%%%%%
		\dfrac{\delta \mathcal{L}'_G }{\delta A^{ab}}
		%%%%%%%%
		&=
		\dfrac{ek}{8\pi^2\ell\ell'} \bigl(T^c + \dfrac{1}{2} ( \dfrac{dk}{k} - \dfrac{d\ell}{\ell} - \dfrac{d\ell'}{\ell'} ) \wedge \beta^c \bigr) \wedge \beta^d \varepsilon_{abcd} \nonumber \\
		&+ \dfrac{k}{2\pi^2\ell\ell'} \Bigl( y T_a + \dfrac{y}{2} ( \dfrac{dk}{k} - \dfrac{d\ell}{\ell} - \dfrac{d\ell'}{\ell'} ) \wedge \beta_a - v \lambda \wedge \beta_a \Bigr)  \wedge \beta_b
		%%%%%%%%%
		= \left\{
		\begin{array}{ll}
			\dfrac{1}{2} \mathfrak{s}_{ab} \text{ for } S_M \text{ or } \tS_M \\
			\dfrac{\tk^2}{2(\tl\tl')^2} \mathfrak{s}_{ab} \text{ for } S'_M \text{ or } \tS'_M 
		\end{array}	
		\right. \label{eom spin Möbius quasi} \\
		%%%%%%%%%%%%%%%%%%%%%%%
		\dfrac{\delta \mathcal{L}'_G}{\delta \lambda} & = 
		- \dfrac{vk}{4\pi^2\ell\ell'} \Tmod[\mu \nu]{a} \beta_{a,\rho} \varepsilon^{\gamma \mu \nu \rho }
		= 
		\left\{
		\begin{array}{ll}
			0 \text{ for } S_M \text{ or } S'_M \\
			\dfrac{1}{12} h^\varepsilon\bigl( i \psi, \eta_{ab} \gamma^a \betamod[b]{\mu} \delta^{\mu\gamma} \psi \bigl) \text{ for } \tS_M \\
			\dfrac{\tk^2}{12(\tl\tl')^2} h^\varepsilon\bigl( i \psi, \eta_{ab} \gamma^a \betamod[b]{\mu} \delta^{\mu\gamma} \psi \bigl) \text{ for } \tS'_M
		\end{array}	
		\right.
		\numberthis \label{eom scalar quasi} 
	\end{align}
	with \eqref{eom spin Möbius quasi} linking the trace part of torsion $\Tmod[rc]{c}$, dilations $\lambda$,
the parameters $k$, $\ell$, $\ell'$ and spin density $\mathfrak{s}$ according to the relation:\vspace{-0.3cm}
\begin{align}
	\label{link torsion spin density Mobius}
	\overbrace{T^c_{~ab} \varepsilon^{ab}_{~~cr}}^{-6\tilde{t}_r} + \dfrac{y}{e} \bigl( 8 \Tmod[rc]{c} + 12 \partial_r (\dfrac{k}{\ell\ell'}) \dfrac{\ell \ell'}{k} \bigr) - \dfrac{24 v}{e} \lambda_r = \left\{
	\begin{array}{ll}
		\dfrac{4\pi^2\ell\ell'}{ek} \mathfrak{s}^{ab}_{~~abr} \text{ for } S_M \text{ or } \tS_M \\
		\dfrac{\tk}{\tl\tl'}\dfrac{4\pi^2\ell_0\ell_0'}{ek_0}  \mathfrak{s}^{ab}_{~~abr} \text{ for } S'_M \text{ or } \tS'_M
	\end{array}	
	\right.
\end{align}

\vspace{-0.2cm}We can observe that unless we choose a matter action involving dilational degrees of freedom (DOF) such as $\tS_M$ or $\tS'_M$ \eqref{eom scalar quasi} either implies $v=0$ or $T^a\wedge \beta_a=0= \tilde{t^a}$ (null axial torsion).

And the EOM's relative to $\ell$, $\ell'$ and $k$ are respectively: \vspace{-0.5cm}
{\footnotesize
	\begin{align}
		\hspace{-0.5cm}
		- \dfrac{k}{\ell^2 \ell'} \biggl(
		\dfrac{e}{16\pi^2 } ( R^{ab} + \dfrac{k}{\ell\ell'} \beta^a \wedge \beta^b ) \wedge \beta^c \wedge \beta^d \varepsilon_{abcd}
		%%%%%%%%%
		+ \dfrac{1}{4 \pi^2}
		\bigl(
		y R^{ab} \wedge \beta_a \wedge \beta_b 
		-2v \lambda \wedge \beta^a \wedge T_a
		\bigr)
		%%%%%%%%%
		\biggr)
		= \left\{
		\begin{array}{ll}
			0 \text{ for } S_M \text{ or } \tS_M \\
			\dfrac{2 \tk^2}{\ell_0 \tl^3 \tl'^2} \mathcal{L}_M \text{ for } S'_M \\
			\dfrac{2 \tk^2}{\ell_0 \tl^3 \tl'^2} \tilde{\mathcal{L}}_M \text{ for } \tS'_M
		\end{array}	
		\right. 
		\label{eom l Möbius}  \\ \hspace{-2cm}
		%%%%%%%%%%%%%%%%%%%%%%%%%%%%%%%%%%%%%%%%%%%%%%%%%%%%%%%%%%%%%
		- \dfrac{k}{\ell \ell'^2} \biggl(
		\dfrac{e}{16\pi^2 } ( R^{ab} + \dfrac{k}{\ell\ell'} \beta^a \wedge \beta^b ) \wedge \beta^c \wedge \beta^d \varepsilon_{abcd}
		%%%%%%%%%
		+ \dfrac{1}{4 \pi^2}
		\bigl(
		y R^{ab} \wedge \beta_a \wedge \beta_b
		-2v \lambda \wedge \beta^a \wedge T_a
		\bigr)
		%%%%%%%%%
		\biggr)
		= \left\{
		\begin{array}{ll}
			0 \text{ for } S_M \text{ or } \tS_M \\
			\dfrac{2 \tk^2}{\ell'_0 \tl^2 \tl'^3} \mathcal{L}_M \text{ for } S'_M \\
			\dfrac{2 \tk^2}{\ell'_0 \tl^2 \tl'^3} \tilde{\mathcal{L}}_M \text{ for } \tS'_M
		\end{array}	
		\right. 
		\label{eom l' Möbius} \\ \hspace{-2cm}
		%%%%%%%%%%%%%%%%%%%%%%%%%%%%%%%%%%%%%%%%%%%%%%%%%%%%%%%%%%
		\dfrac{1}{\ell\ell'} \biggl(
		\dfrac{e}{16\pi^2 } ( R^{ab} + \dfrac{k}{\ell\ell'} \beta^a \wedge \beta^b ) \wedge \beta^c \wedge \beta^d \varepsilon_{abcd}
		%%%%%%%%%
		+ \dfrac{1}{4 \pi^2}
		\bigl(
		y R^{ab} \wedge \beta_a \wedge \beta_b 
		-2v \lambda \wedge \beta^a \wedge T_a
		\bigr)
		%%%%%%%%%
		\biggr)
		= \left\{
		\begin{array}{ll}
			0 \text{ for } S_M \text{ or } \tS_M \\
			\dfrac{2 \tk}{k_0 (\tl \tl')^2} \mathcal{L}_M \text{ for } S'_M \\
			\dfrac{2 \tk}{k_0 (\tl \tl')^2} \tilde{\mathcal{L}}_M \text{ for } \tS'_M
		\end{array}	
		\right.
		\label{eom k Möbius}
	\end{align}
}%

\vspace{-0.2cm}Simplifying the system of EOM's and identifying $ e = \dfrac{3\pi}{2\Lambda_0 G_0} $, $ \Lambda_0 = - (\dfrac{3k }{ \ell \ell'\ })|_{x_0} = \dfrac{\chi_0}{4} $ as well as $\gamma = \dfrac{e}{r+2y}$,  we obtain: \vspace{-0.0cm}
{\footnotesize
	\begin{align} 
		&\hspace{-0.0cm}
		G_{kc} 
		+ \dfrac{\chi}{4} \eta_{ck}
		+ (\dfrac{\Lambda_0 G_0r}{3\pi} - \dfrac{1}{2\gamma}) \Rmod[rs]{ab} \varepsilon^{rs}_{~~bk} \delta_{ac} 
		- \dfrac{4\Lambda_0G_0v}{3\pi} \Bigl( \lambda_b T_{c,rs} + %2 
		\bigl( \dfrac{d_r\chi}{\chi}
		+ d_r\bigr) \lambda_b \beta_{c,s}   \Bigr) \varepsilon^{brs}_{~~~k} \nonumber \\
		%%%%
		&\hspace{-0.0cm} = \left\{
		\begin{array}{ll}
			\dfrac{8 \pi G_0}{\tilde{\chi}}   \tau_{ck} \text{ for } S_M \\
			8 \pi G_0 \tilde{\chi} \tau_{ck} \text{ for } S'_M \\
			\dfrac{8 \pi G_0}{\tilde{\chi}} \bigl(\tau_{ck}  + \dfrac{1}{2} \sqrt{|\tg| }h^\varepsilon\bigl( i \psi, \gamma^a \eta_{ac} \lambda_\mu \psi \betainvmod[\mu]{k} ) \bigr) \text{ for } \tS_M \\
			8 \pi G_0 \tilde{\chi} \bigl(\tau_{ck}  + \dfrac{1}{2} \sqrt{|\tg| }h^\varepsilon( i \psi, \gamma^a \eta_{ac} \lambda_\mu \psi \betainvmod[\mu]{k} ) \bigr) \text{ for } \tS'_M
		\end{array}	
		\right.
		\label{eom frame Möbius '2}  \\[2mm] \hspace{-2cm}
		%%%%%%%%%%%%%%%%%%%%%%%%%%%%%%%%%%%%%%%%%%%%%%%%%%%%%%%%%
		&\hspace{-0.0cm}\bigl(T^c + \dfrac{1}{2} \dfrac{d\chi}{\chi} \wedge \beta^c \bigr) \wedge \beta^d \varepsilon_{abcd} 
		+ 4  \Bigl( (\dfrac{1}{2\gamma} - \dfrac{\Lambda_0 G_0r}{3\pi}) \bigl(  T_a + \dfrac{1}{2}\dfrac{d\chi}{\chi} \wedge \beta_a \bigr) - \dfrac{2\Lambda_0G_0v}{3\pi} \lambda \wedge \beta_a \Bigr) \wedge \beta_b \nonumber \\
		%%%%%%%%%
		&\hspace{-0.0cm}= \left\{
		\begin{array}{ll}
			-\dfrac{8 \pi G_0}{\tilde{\chi}}  \mathfrak{s}_{ab} \text{ for } S_M \text{ or } \tS_M \\
			- 8 \pi G_0 \tilde{\chi} \mathfrak{s}_{ab} \text{ for } S'_M \text{ or } \tS'_M
		\end{array}	
		\right.  \label{eom spin Möbius '2} \\
		%%%%%%%%%%%%%%%%%%%%%%%%%%%%%%%%%%%%%%%%%%%%%%%%%%%%%%%%%%
		&\hspace{-0.0cm} \dfrac{18v}{\pi^2} \chi \tilde{t}^a \betainvmod[\gamma]{a}= \dfrac{3v}{\pi^2} \chi \Tmod[\mu \nu]{a} \beta_{a,\rho} \varepsilon^{\gamma \mu \nu \rho }
		= 
		\left\{
		\begin{array}{ll}
			0 \text{ for } S_M \text{ and } S'_M \\
			\dfrac{1}{2} h^\varepsilon\bigl( i \psi, \eta_{ab} \gamma^a \betamod[b]{\mu} \delta^{\mu\gamma} \psi \bigl) \text{ for } \tS_M \\
			\dfrac{\tk}{2\tl\tl'} h^\varepsilon\bigl( i \psi, \eta_{ab} \gamma^a \betamod[b]{\mu} \delta^{\mu\gamma} \psi \bigl) \text{ for } \tS'_M
		\end{array}	
		\right. \label{EOM lambda tilde 2} \\
		%%%%%%%%%%%%%%%%%%%%%%%%%%%%%%%%%%%%%%%%%%%%%%%%%%%%%%%%%%
		&\hspace{-0.0cm} \chi = - \dfrac{12 k}
		{\ell \ell'} = 
		\left\{
		\begin{array}{ll}
			\chi_1 = \rR + \dfrac{y}{e} \Rmod[\mu\nu]{ab} \varepsilon^{\mu\nu}_{~~ab} -\dfrac{2v}{e} \lambda_\mu \beta^a_\nu T_{a,\rho\lambda} \varepsilon^{\mu\nu\rho\lambda} \text{ for } S_M \text{ or } \tS_M \\
			\chi_2 = \chi_1 / \bigl( 1 - \dfrac{288 \pi^2}{e\chi_0'^2} \sqrt{|\tg|} h^\varepsilon( i \psi, \eta_{ab} \gamma^a \beta^b_\mu D^\mu\psi)\bigr) \text{ for } S'_M \\
			\chi_3 = \chi_1 / \bigl( 1 - \dfrac{288 \pi^2}{e\chi_0'^2} \sqrt{|\tg|} h^\varepsilon( i \psi, \eta_{ab} \gamma^a \beta^b_\mu D^\mu_0\psi)\bigr) \text{ for } \tS'_M  \label{eom csm cst 1} 
		\end{array}	
		\right.
	\end{align}
}%
	Now, the choice of taking $\tS_M$ (or $\tS'_M$) over $S_M$ (or $S'_M$), mainly changes relations \eqref{eom frame Möbius '2} and \eqref{EOM lambda tilde 2}, respectively adding $\dfrac{1}{2} \sqrt{|\tg| }h^\varepsilon( i \psi, \gamma^a \eta_{ac} \lambda_\mu \psi \betainvmod[\mu]{k} )$ as a source for curvature and linking torsion to matter fields via the axial vector $\tilde{t}^a = \dfrac{1}{6} T^b_{~\mu \nu} \varepsilon^{a~\mu \nu}_{~b}$. 
	Going from $S'_M$ to $\tS'_M$ also changes $D$ to $D_0$ in \eqref{eom csm cst 1}.
	
	One may also observe that similarly to what was done in \cite{thibaut2024gravity}, using \eqref{link torsion spin density Mobius} in \eqref{eom frame Möbius '2} replaces the dilations in favor of a new secondary source of curvature expressed in terms of the spin density of matter, torsion, their variations and the ratio $\dfrac{d\Lambda_G}{\Lambda_G} = \dfrac{d\chi}{\chi}$. The expression of this secondary source of curvature being:
	\begin{align}
	 	\dfrac{\Lambda_0v}{6\pi^2} \bigl( \lambda_b T_{c,rs} + 2( \dfrac{\partial_r\Lambda_G}{\Lambda_G} + \partial_r ) \lambda_b \beta_{c,s} \bigr)\varepsilon^{brs}_{~~~k}
	\end{align}
	with
	\begin{align}
	    \lambda_r = 
	    \left\{
		\begin{array}{ll}
			\dfrac{1}{v} \bigl( \dfrac{\pi}{16 \Lambda_0 G_0} T^c_{~ab} \varepsilon^{ab}_{~~cr} +  \dfrac{y}{3} ( T^c_{~rc} + \dfrac{3}{2} \dfrac{\partial_r\Lambda_G}{\Lambda_G} ) + \dfrac{\pi^2}{2\Lambda_0\tilde{\chi}} \mathfrak{s}^{ab}_{~~abr} \bigr) \text{ for } S_M \text{ or } \tS_M \\
			\dfrac{1}{v} \bigl( \dfrac{\pi}{16 \Lambda_0 G_0} T^c_{~ab} \varepsilon^{ab}_{~~cr} + \dfrac{y}{3} ( T^c_{~rc} + \dfrac{3}{2} \dfrac{\partial_r\Lambda_G}{\Lambda_G} ) + \dfrac{\pi^2\tilde{\chi}}{2\Lambda_0}  \mathfrak{s}^{ab}_{~~abr} \bigr) \text{ for } S'_M \text{ or } \tS'_M 
		\end{array}	
		\right.
	\end{align}
	In the following table we compare the behavior of $\Lambda_G$, $\Lambda_{\text{vac}}$ and $G$ depending on the choice of matter action:
	\begin{align} \label{summary quasi Mobius} 
		\left\{
		\begin{array}{ll}
			S_T \text{ or } \tS_M &\Leftrightarrow \Lambda_G = \dfrac{\chi}{4} \text{ and }
			\left\{
			\begin{array}{ll}
				\Lambda_{\text{vac}} \\
				G
			\end{array}
			\right\} \propto \dfrac{1}{\chi} \\[2mm]
			S'_T \text{ or } \tS'_M  &\Leftrightarrow \Lambda_G = \dfrac{\chi}{4}  \text{ and }
			\left\{
			\begin{array}{ll}
				\Lambda_{vac} = 8 \pi G \rho_{vac}  \\
				G=G_0 \tilde{\chi}
			\end{array}
			\right\} \propto \chi
		\end{array}
		\right.
	\end{align}
	
	Just as in section \ref{Moving Lorentzian Cartan geometry} we derive that choosing a total action $S = S_G + S^f_M$ or $S = S_G + \tilde{S}^f_M$ (see \eqref{f matter action Mobius}) with $f(k,k',\ell) = \dfrac{\tk \tk'}{\tl}$ allows to retrieve a spacetime-dependent dark energy $\Lambda_G$ with the usual Newton constant $G_0$ and vacuum contribution $\Lambda_{vac} = 8 \pi G_0 \rho_{vac}$.
	
	In this section we showed the moving Möbius geometry made it possible to recover conformal gravity tensors (Weyl, Schouten, Cotton) or the previous moving Lorentzian geometry model with added dilational contributions by respectively considering either the normal Möbius geometry ($Ricci(\bOm) = 0$ and $\bOmp=0$) or tetrads obeying the $\alpha = \dfrac{k}{2} \beta$ condition. 
	Conformal gravity \cite{mannheim2012making} and Lorentzian models (Poincaré/dS/AdS) have their own strong points and caveats in describing gravity and its different observable phenomena.
	The next step shall therefore be to look for a choice of tetrads $(\alpha,\beta)$ that generalizes the previous cases and allows to recover transitions between both models, similarly to the possibility of a crossing from a positive to a negative (or vice versa) $\Lambda_G$ described in the first part of this article.
	
	We also studied the addition of a term combining both dilations and torsion of the form $\dfrac{v k}{2\pi^2\ell \ell'} \eta_{ab} T^a \wedge \beta^b \wedge \lambda $ to the action. This last term leads to a secondary source of curvature in the modified Einstein equations and most importantly modifies the expression of $\chi$ and therefore of $\Lambda_G$ and $G$ by adding a contribution $-\dfrac{2v}{e} \lambda_\mu \beta^a_\nu T_{a,\rho\lambda} \varepsilon^{\mu\nu\rho\lambda}$.

		\section{Conclusion}
		
		We have first introduced the notion of moving geometries described by quotients of Lie groups and algebras with spacetime dependant structure "fields" in place of structure constants and built a deformed topological action functional for moving Lorentzian and Lorentz$\times$Weyl geometries. 
		
 		We then showed that considering moving geometries naturally leads via the action principle to a theory with cosmological ($\Lambda_G$) and Newton ($G$) fields in place of constants. These fields may in fact be interpreted as sources of dark energy and dark matter with differing variations and expressions in terms of scalars of the spacetime curvature, torsion and matter fields depending on the choice of matter action. In particular cases we even retrieve existing dark energy models such as the Ricci dark energy model of \cite{gao_holographic_2009} with the added twist of a varying gravitational coupling, giving a clear interpretation of this model (with parameter $\alpha = 1/2$) in terms of moving geometries. Hence leading to a novel way of building theories with evolving dark energy and gravitational coupling. The main advantage is that the scalar fields are entirely determined by the action principle at each point.
 		
 		We summarize in the diagram below the different geometries studied in this article and the construction of the gauge and matter actions:
 		
 		\begin{figure}
			\hspace{-1.75cm}\includegraphics[scale=0.45]{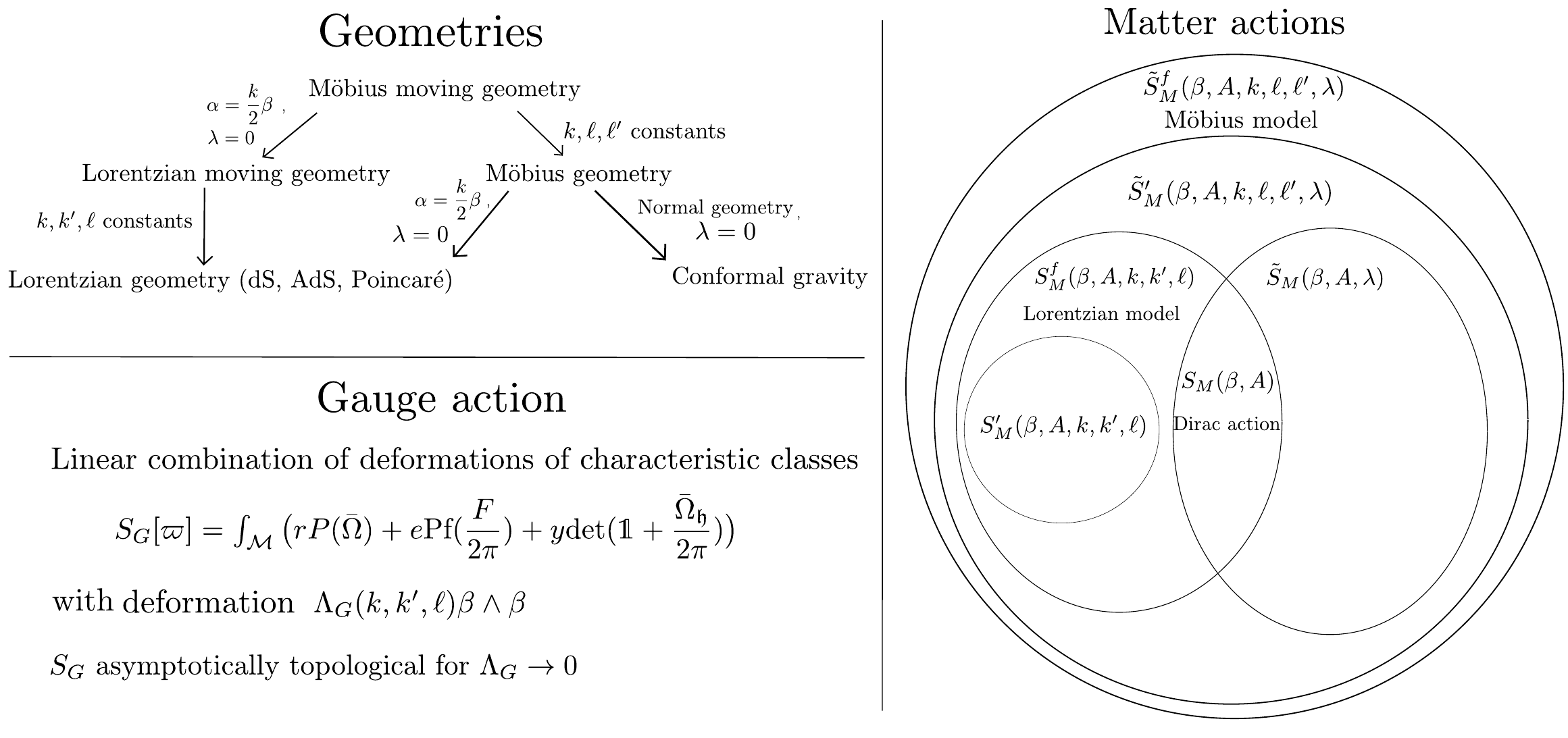}%\caption{Recapitulative diagram of the different geometries and the gauge and matter actions studied.}
 		\end{figure}
 		
 		We also consign in the diagram below the different expressions of $\Lambda_G$, $G$ and the fields $\chi_i$ depending on the choice of matter action\footnote{One may see the columns related to $S_M$ and $\tS_M$ lead to the same results for  $\Lambda_G,G$ and $\chi$, however in the $\tS_M$ and $\tS'_M$ cases a major difference is that the energy-momentum tensor $\tau$ now acquires a contribution linked to the dilations $\lambda$ as can be seen in \eqref{eom frame Möbius '2}.}:
 		\begin{center}
 			\begin{tabular}{ c | c | c | c | c }
 				& $S_M$ & $\tS_M$  & $S'_M$ & $\tS'_M$ \\ \hline
 				$\Lambda_G \propto $ & $\chi_1$ & $\chi_1$ & $\chi_2$ & $\chi_3$ \\ \hline
 				$G \propto $ & $\dfrac{1}{\chi_1}$ & $\dfrac{1}{\chi_1}$ & $\chi_2$ & $\chi_3$  \\  \hline
 				$\chi$ & $\chi_1 = \rR + \dfrac{y}{e} \Rmod[\mu\nu]{ab} \varepsilon^{\mu\nu}_{~~ab} -\dfrac{2v}{e} \lambda_\mu \beta^a_\nu T_{a,\rho\lambda} \varepsilon^{\mu\nu\rho\lambda}$ & $\chi_1 $ & $\chi_2 = \dfrac{\chi_1}{1-\mathcal{L}_M}$ & $\chi_3 = \dfrac{\chi_1}{1-\tilde{\mathcal{L}}_M}$  \\
 			\end{tabular}
 	%	\caption{Table 1: Table summarizing the expressions of the cosmological field $\Lambda_G$ and the gravitational coupling field $G$  for the different matter actions.}
 		\end{center}
 		Recent observations and analysis made in \cite{sengupta2025cosmological} tend to favor the minimal coupling matter action $S_M$ with a $G\propto 1/\chi$ growing over time over the choice $S'_M$ where $G' \propto \chi'$ instead decays.
 		
		Simplified models ($r=y=0$) with a minimal coupling to matter via the Dirac action $S_M$ lead to a source of dark energy $\Lambda_G \propto \rR$ proportional to the Ricci scalar that decreases over cosmological time in favor of a growing gravitational coupling $G\propto \dfrac{1}{\rR}$. Further study shall determine if this type of behavior could in fact alleviate cosmological tensions as noted in \cite{tang_uniting_2025,collaboration_desi_2025} and potentially provide some kind of Lenz law by decoupling matter from curvature in the high curvature limit. The inverse dependence of $G$ on $\rR$ should also lead to higher values of $G$ in low-curvature regions, thus for example leading to a higher $G$ in the outermost region of galaxies when compared to the innermost region. This is in fact reminiscent of other modified theories of gravity such as MOND \cite{famaey2012modified}, MOG/Scalar-tensor-vector gravity \cite{moffat2006scalar}.
		These behaviors should be investigated in more details in the future to see if they could help describe galaxy rotation curves.
		
		We also found the parameter $y$ acts as a coupling between the pseudo curvature scalar and non radiative matter ($w\neq 1/3$) when one considers a perfect fluid approximation of matter.
		
		Lastly the requirement of having $\Lambda_G \rightarrow 0$ 
		in order to have an asymptotically topological theory already detailed in \cite{thibaut2024gravity} thus motivates a perturbative expansion in terms of $\Lambda_G$.
		
		The next logical steps should consist of a more detailed study of these models, notably by investigating the stability of solutions of the equations of motion under perturbations, and analysing the impact of the torsion and matter fields in the general case ($r\neq 0 \neq y$ and with matter actions $S^f_M$ and $\tS^f_M$).

        	\subsubsection*{Acknowledgements}
	
	           I would like to thank Serge Lazzarini and Thierry Masson for helpful discussions as well as the reviewers for their helpful comments.

		\addcontentsline{toc}{section}{References}

		\bibliographystyle{unsrturl}

	\end{document}